\documentclass[a4paper]{article}

\usepackage{amsthm}
\usepackage{amssymb}
\usepackage{booktabs}
\usepackage{color}
\usepackage{colortbl}
\usepackage{enumitem}
\usepackage[hang,flushmargin]{footmisc}
\usepackage[margin=25mm]{geometry}
\usepackage[
	pdftitle={Generalisations of the Harer-Zagier recursion for 1-point functions},
	pdfauthor={Anupam Chaudhuri and Norman Do},
	ocgcolorlinks,
	linkcolor=linkblue,
	citecolor=linkred,
	urlcolor=linkred]
		{hyperref}
\usepackage{mathpazo}
\usepackage{mathtools}
\usepackage{microtype}
\usepackage{sectsty}
	\sectionfont{\large}
	\subsectionfont{\normalsize}
\usepackage{tabularx}
\usepackage{titlesec}
	\titlespacing{\section}{0pt}{12pt}{0pt}
	\titlespacing{\subsection}{0pt}{6pt}{0pt}
\usepackage[capitalise,noabbrev]{cleveref}
	\crefname{equation}{equation}{equations}

\theoremstyle{plain}
	\newtheorem{theorem}{Theorem}
	\newtheorem{proposition}[theorem]{Proposition}
	\newtheorem{corollary}[theorem]{Corollary}
	\newtheorem{lemma}[theorem]{Lemma}
	\newtheorem{conjecture}[theorem]{Conjecture}
	\newtheorem*{conjecture*}{Conjecture}
	\numberwithin{theorem}{section}
\theoremstyle{definition}
	\newtheorem{definition}[theorem]{Definition}
	\newtheorem{example}[theorem]{Example}
\theoremstyle{remark}
	\newtheorem*{remark*}{Remark}

\definecolor{linkred}{rgb}{0.75,0,0}
\definecolor{linkblue}{rgb}{0,0,0.75}
\definecolor{lightblue}{rgb}{0.95,0.95,1}
\definecolor{lightred}{rgb}{1,0.95,0.95}
\definecolor{MapleRed}{rgb}{1,0,0}
\definecolor{MapleBlue}{rgb}{0,0,1}
\definecolor{MaplePink}{rgb}{1,0,1} 

\def\MapleInput#1{\noindent{{\small $>$ {\tt \color{MapleRed}{#1} }}}}
\def\MapleOutput#1{{\begin{center} \color{MapleBlue}{#1} \end{center}}}

\newcommand\blfootnote[1]{
	\begingroup
	\renewcommand\thefootnote{}\footnote{#1}
	\addtocounter{footnote}{-1}
	\endgroup
}

\setlist{nolistsep}

\newcommand {\dd}{\mathrm{d}}
\newcommand {\h}{\hbar}

\begin{document}

{\large \bfseries Generalisations of the Harer--Zagier recursion for 1-point functions}

{\bfseries Anupam Chaudhuri and Norman Do}

{\em Abstract.} Harer and Zagier proved a recursion to enumerate gluings of a $2d$-gon that result in an orientable genus $g$ surface, in their work on Euler characteristics of moduli spaces of curves. Analogous results have been discovered for other enumerative problems, so it is natural to pose the following question: how large is the family of problems for which these so-called 1-point recursions exist?

In this paper, we prove the existence of 1-point recursions for a class of enumerative problems that have Schur function expansions. In particular, we recover the Harer--Zagier recursion, but our methodology also applies to the enumeration of dessins d'enfant, to Bousquet-M\'{e}lou--Schaeffer numbers, to monotone Hurwitz numbers, and more. On the other hand, we prove that there is no 1-point recursion that governs simple Hurwitz numbers. Our results are effective in the sense that one can explicitly compute particular instances of 1-point recursions, and we provide several examples. We conclude the paper with a brief discussion and a conjecture relating 1-point recursions to the theory of topological recursion.
\blfootnote{
\emph{Date:} 30 December 2018 \\ The second author was supported by the Australian Research Council grants DE130100650 and DP180103891.}

~

\hrule

\setlength{\parskip}{0pt}
\tableofcontents
\setlength{\parskip}{8pt}

~

\hrule

\section{Introduction}

For integers $g \geq 0$ and $d \geq 1$, let $a_g(d)$ denote the number of ways to glue the edges of a $2d$-gon in pairs to obtain an orientable genus $g$ surface. The data of a surface constructed by gluing edges of polygons in pairs is often referred to in the literature as a {\em ribbon graph}. In their pioneering work, Harer and Zagier apply matrix model techniques to this enumeration of ribbon graphs with one face to deduce a formula for the virtual Euler characteristics of moduli spaces of curves. One consequence of their calculation is the fact that the numbers $a_g(d)$ satisfy the following recursion~\cite{har-zag}.
\begin{equation} \label{eq:harerzagier}
(d+1) \, a_g(d) = 2(2d-1) \, a_g(d-1) + (2d-1) (d-1) (2d-3) \, a_{g-1}(d-2)
\end{equation}

Despite the simple appearance of this formula, Zagier later stated~\cite{lan-zvo}: ``No combinatorial interpretation of the recursion\ldots\ is known.'' The Harer--Zagier recursion has since attracted a great deal of interest, and there now exist several proofs, some of which are combinatorial in nature~\cite{akh-sha, cha-fer-fus, gou-nic, las, mor-sha, pit}.

In more recent work of the second author and Norbury~\cite{do-nor}, as well as the subsequent work of Chekhov~\cite{che}, an analogue of the Harer--Zagier recursion was deduced for the number of {\em dessins d'enfant} with one face. More precisely, let $b_g(d)$ denote the number of ways to glue the edges of a $2d$-gon, whose vertices are alternately coloured black and white, in pairs to obtain an orientable genus $g$ surface. Of course, we impose the caveat that vertices may only be glued together if they share the same colour. The numbers $b_g(d)$ satisfy the following recursion.
\begin{equation} \label{eq:donorbury}
(d+1) \, b_g(d) = 2(2d-1) \, b_g(d-1) + (d-1)^2 (d-2) \, b_{g-1}(d-2)
\end{equation}

It is natural to embed the problem of calculating $a_g(d)$ into the more general enumeration of ways to glue the edges of $n$ labelled polygons with $d_1, d_2, \ldots, d_n$ sides to obtain an orientable genus $g$ surface. This problem then lends itself naturally to a simple combinatorial recursion, whose roots lie in the work of Tutte~\cite{tut}, but was first expressed by Walsh and Lehman~\cite{wal-leh}. The mechanism for such a recursion comes from removing an edge from the ribbon graph formed by the edges of the polygons, and observing that one is left with either a simpler ribbon graph or the disjoint union of two simpler ribbon graphs. The cost of combinatorial simplicity is the necessity to consider gluings of an arbitrary number of polygons, rather than gluings of just one polygon.

Recursions similar to those expressed in~\cref{eq:harerzagier,eq:donorbury} have appeared in other contexts, such as random matrix theory~\cite{led}. However, it is not true in general that these recursions involve three terms, as in the examples above. In the context of enumerative geometry and mathematical physics, the analogues of $a_g(d)$ and $b_g(d)$ are known as {\em 1-point invariants}, since they often arise as expansion coefficients of 1-point correlation functions. And more generally, the enumeration of ways to glue $n$ polygons to obtain surfaces produces numbers known as {\em $n$-point invariants}. The preceding discussion motivates us to make the following definition.

\begin{definition} \label{def:1pointrecursion}
We say that the collection of numbers $n_g(d) \in \mathbb{C}$ for integers $g \geq 0$ and $d \geq 1$ satisfies a {\em 1-point recursion} if there exist integers $i_{\max}$, $j_{\max}$ and complex polynomials $p_{ij}$, not all equal to zero, such that
\begin{equation} \label{eq:1pointrecursion}
\sum_{i=0}^{i_{\max}} \sum_{j=0}^{j_{\max}} p_{ij}(d) \, n_{g-i}(d-j) = 0,
\end{equation}
whenever all terms in the equation are defined.
\end{definition}

The current work is motivated by the following interrelated questions.
\begin{itemize}
\item What unified proofs of 1-point recursions exist, which encompass both \cref{eq:harerzagier,eq:donorbury}?
\item How universal is the the notion of a 1-point recursion?
\end{itemize}

We partially answer these questions by first observing that the enumeration of both ribbon graphs and dessins d'enfant can be expressed in terms of Schur functions. This suggests that 1-point recursions may exist more generally for problems that may be defined in an analogous way. Thus, we consider {\em double Schur function expansions} of the following form.
\begin{align} \label{eq:partitionfunction}
Z(\mathbf{p}; \mathbf{q}; \h) &= \sum_{\lambda \in {\mathcal P}} s_\lambda(p_1, p_2, \ldots) \, s_\lambda(\tfrac{q_1}{\h}, \tfrac{q_2}{\h}, \ldots) \, \prod_{\Box \in \lambda} G(c(\Box) \h) \notag \\
&= \exp \bigg[ \sum_{g=0}^\infty \sum_{n=1}^\infty \sum_{d_1, d_2, \ldots, d_n = 1}^\infty N_{g,n}(d_1, d_2, \ldots, d_n) \, \frac{\h^{2g-2+n}}{n!} \, p_{d_1} p_{d_2} \cdots p_{d_n} \bigg]
\end{align}
The precise meaning of all terms appearing in this equation will be discussed in \cref{sec:problems}. It currently suffices to observe that the ``enumerative problem'' is stored in the numbers $N_{g,n}(d_1, d_2, \ldots, d_n)$ appearing in the second line. These numbers have been recently studied in the work of Alexandov, Chapuy, Eynard and Harnad~\cite{ale-cha-eyn-har}, where they are given a combinatorial interpretation and referred to as {\em weighted Hurwitz numbers}.

The primary contribution of this paper is an approach to proving 1-point recursions for such ``enumerative problems''. In particular, our main result is the following.

\begin{theorem} \label{thm:main}
Let $G(z) \in \mathbb{C}(z)$ be a rational function with $G(0) = 1$ and suppose that finitely many terms of the sequence $q_1, q_2, q_3, \ldots$ of complex numbers are non-zero. Then the numbers $n_g(d) = d \, N_{g,1}(d)$ defined by \cref{eq:partitionfunction} satisfy a 1-point recursion.
\end{theorem}

The proof of this theorem will be taken up in \cref{sec:recursion}, where we use the theory and language of holonomic sequences and functions. The basic observation is \cref{lem:holonomic}, which states that a 1-point recursion exists for $n_g(d)$ if and only if the sequence $n_d = \displaystyle\sum_g n_g(d) \, \h^{2g-1}$ is holonomic over $\mathbb{C}(\h)$.

\begin{example}
If we take $G(z) = 1+z$ and $\mathbf{q} = (0, 1, 0, 0, \ldots)$ in \cref{eq:partitionfunction}, then we recover the enumeration of ribbon graphs introduced earlier. In other words, we have $n_g(d) = a_g(d)$, so \cref{thm:main} asserts the existence of a 1-point recursion for the numbers $a_g(d)$.

Analogously, if we take $G(z) = (1+z)^2$ and $\mathbf{q} = (1, 0, 0, \ldots)$ in \cref{eq:partitionfunction}, then we recover the enumeration of dessins d'enfant introduced earlier. In other words, we have $n_g(d) = b_g(d)$, so \cref{thm:main} asserts the existence of a 1-point recursion for the numbers $b_g(d)$. 
\end{example}

One of the features of the theory of holonomic sequences and functions is that there are readily available algorithms to carry out computations, such as those found in the \texttt{gfun} package for {\sc Maple}~\cite{sal-zim}. Our proof of \cref{thm:main} not only asserts the existence of 1-point recursions, but also yields an algorithm to produce them. We use this to determine explicit 1-point recursions for:
\begin{itemize}
\item the enumeration of 3-hypermaps and 3-BMS numbers (see \cref{prop:hypermapsbms}); and

\item the enumeration of monotone Hurwitz numbers (see \cref{prop:monotonehurwitz}).
\end{itemize}

\begin{example}
The monotone Hurwitz numbers satisfy the following 1-point recursion.
\[
d \, m_g(d) = 2(2d-3) \, m_g(d-1) + d(d-1)^2 \, m_{g-1}(d)
\]
\end{example}

As a partial converse to \cref{thm:main}, we prove that there are enumerative problems governed by double Schur function expansions that do not satisfy a 1-point recursion. Of particular note is the case of simple Hurwitz numbers, which arise from \cref{eq:partitionfunction} by taking $G(z) = \exp(z)$ and $\mathbf{q} = (1, 0, 0, \ldots)$.

\begin{proposition}
The simple Hurwitz numbers do not satisfy a 1-point recursion.
\end{proposition}

Underlying our work are the related notions of integrability and topological recursion. Regarding the former, we only remark that the double Schur function expansions of \cref{eq:partitionfunction} are examples of hypergeometric tau-functions for the Toda integrable hierarchy~\cite{orl-shc}. The topological recursion can be used to produce enumerative invariants from a spectral curve, which is essentially a plane algebraic curve satisfying some mild conditions and equipped with certain extra data. From the work of Alexandrov, Chapuy, Eynard and Harnad~\cite{ale-cha-eyn-har}, we know that the assumptions of \cref{thm:main} lead to numbers $N_{g,n}(d_1, d_2, \ldots, d_n)$ in \cref{eq:partitionfunction} that can be calculated via the topological recursion. Furthermore, the associated spectral curve is an explicit rational curve, which depends on the particular choice of $G(z)$ and~$\mathbf{q}$. Combining \cref{thm:main} with the aforementioned work of Alexandrov et al. suggests the following conjecture, whose precise statement will later appear as \cref{con:mainprecise}.

\begin{conjecture} \label{con:main}
Topological recursion on a rational spectral curve produces invariants that satisfy a 1-point recursion. 
\end{conjecture}

In practice, one may only be interested in 1-point functions, as is the case for the problem originally studied by Harer and Zagier~\cite{har-zag}. Calculating these via the topological recursion requires the knowledge of $n$-point functions for all positive integers $n$. Thus, a 1-point recursion can provide an effective tool for calculation, from both the practical and theoretical perspectives. For instance, a 1-point recursion can lead to direct information regarding the structure of 1-point invariants --- see \cref{cor:monotonehurwitz} for an example of this phenomenon.

We conclude the paper with some evidence towards the conjecture above as well as a brief discussion on the related notion of quantum curves. In the context of the double Schur function expansions studied in this paper, quantum curves arise from a specialisation of \cref{eq:partitionfunction} that reduces the summation over all partitions to a summation over 1-part partitions. On the other hand, we will observe that 1-point recursions arise from a different specialisation that reduces it to a summation over hook partitions.

The structure of the paper is as follows.
\begin{itemize}
\item In \cref{sec:problems}, we introduce four classes of enumerative problems that will provide motivation for and examples of our main results. These are: the enumeration of ribbon graphs and dessins d'enfant; Bousquet-M\'{e}lou--Schaeffer numbers; Hurwitz numbers; and monotone Hurwitz numbers. A common thread between these problems is that their so-called partition functions have double Schur function expansions.

\item In \cref{sec:schur}, we precisely define double Schur function expansions and deduce an expression for their 1-point invariants. We also present certain evaluations of Schur functions that will subsequently prove useful.

\item In \cref{sec:recursion}, we recall the notion of holonomicity and relate it to the existence of 1-point recursions. This is used to prove \cref{thm:main} on the existence of 1-point recursions, which then leads to an algorithm for 1-point recursions.

\item In \cref{sec:applications}, we return to the four classes of enumerative problems introduced in \cref{sec:problems}. For three of these, we present examples of 1-point recursions, but for the case of simple Hurwitz numbers, we prove that no such recursion exists. We also demonstrate how 1-point recursions can be used to prove structural results, and sometimes explicit formulas, for 1-point invariants.

\item In \cref{sec:discussion}, we discuss relations between our work and the theory of topological recursion. In particular, we formulate a precise statement of \cref{con:main}, which loosely states that is a 1-point recursion for the invariants arising from topological recursion applied to a rational spectral curve. Some evidence toward this conjecture is presented, along with some remarks on the similarity between our calculation of 1-point recursions and the calculation of quantum curves.
\end{itemize}

\section{Enumerative problems} \label{sec:problems}

Our work is primarily motivated by the Harer--Zagier formula for the enumeration of ribbon graphs with one face~\cite{har-zag}, as well as the analogue for the enumeration of dessins d'enfant with one face~\cite{che,do-nor}. Apart from the obvious similarities between these two problems, they also both arise from double Schur function expansions. So we propose to study the broad class of ``enumerative problems'' stored in double Schur function expansions of the general form
\[
Z(\mathbf{p}; \mathbf{q}; \h) = \sum_{\lambda \in {\mathcal P}} s_\lambda(p_1, p_2, \ldots) \, s_\lambda(\tfrac{q_1}{\h}, \tfrac{q_2}{\h}, \ldots) \, F_\lambda(\h).
\]
Here, $\mathcal P$ denotes the set of all partitions (including the empty partition), $s_\lambda(p_1, p_2, \ldots)$ denotes the Schur function expressed in terms of power sum symmetric functions, and $F_\lambda(\h)$ is a formal power series in $\h$ for each partition $\lambda$. We use the shorthand $\mathbf{p} = (p_1, p_2, p_3, \ldots)$ and $\mathbf{q} = (q_1, q_2, q_3, \ldots)$ throughout the paper. Following the mathematical physics literature, we will refer to such power series as {\em partition functions} (although we note that this name does not refer to the integer partitions that appear in the equation above).

For our applications, we will take $F_\lambda(\h)$ to have the so-called {\em content product} form
\[
F_\lambda(\h) = \prod_{\Box \in \lambda} G(c(\Box) \h).
\]
Here, the product is over the boxes in the Young diagram for $\lambda$, $G(z) \in \mathbb{C}[[z]]$ is a formal power series normalised to have constant term 1, and $c(\Box)$ denotes the content of the box. Recall that the {\em content} of a box in row $i$ and column $j$ of a Young diagram is the integer $j-i$.

The partition function can be expressed as 
\[
Z(\mathbf{p}; \mathbf{q}; \h) = \exp \bigg[ \sum_{g=0}^\infty \sum_{n=1}^\infty \sum_{d_1, d_2, \ldots, d_n = 1}^\infty N_{g,n}(d_1, d_2, \ldots, d_n) \, \frac{\h^{2g-2+n}}{n!} \, p_{d_1} p_{d_2} \cdots p_{d_n} \bigg],
\]
where $N_{g,n}(d_1, d_2, \ldots, d_n) \in \mathbb{C}[q_1, q_2, \ldots]$. For various natural choices of the formal power series $G(z)$ and the weights $q_1, q_2, q_3, \ldots$, the quantity $N_{g,n}(d_1, d_2, \ldots, d_n)$ enumerates objects of combinatorial interest. We will be primarily concerned with the 1-point invariants that arise when $n = 1$. In particular, we consider the numbers $n_g(d) = d \, N_{g,1}(d)$, with the goal of determining whether or not there exists a 1-point recursion governing these numbers.

We now proceed to examine four classes of combinatorial problems that arise from double Schur function expansions. Readers looking for the general description of double Schur function expansions and their 1-point recursions may wish to skip directly to \cref{sec:schur}.

\subsection{Ribbon graphs and dessins d'enfant}

A ribbon graph --- also known as a map, embedded graph, fat graph or rotation system --- can be thought of as the 1-skeleton of a cell decomposition of an oriented compact surface. Ribbon graphs arise naturally in various areas of mathematics, including topological graph theory, moduli spaces of Riemann surfaces, and matrix models~\cite{lan-zvo}. A more formal definition is the following.

\begin{definition}
A {\em ribbon graph} is a finite connected graph equipped with a cyclic ordering of the half-edges meeting at each vertex. An {\em isomorphism} between two ribbon graphs is a bijection between their sets of half-edges that preserves all adjacencies, as well as the cyclic ordering of the half-edges meeting at each vertex.
\end{definition}

The underlying graph of a ribbon graph is precisely the 1-skeleton of a cell decomposition of a compact connected orientable surface. The cyclic ordering of the half-edges meeting at every vertex allows one to reconstruct the 2-cells and hence, the underlying oriented compact surface. Thus, one can assign a genus to a ribbon graph.

Alternatively, one can encode a ribbon graph as a pair of permutations $(\tau_0, \tau_1)$ such that $\tau_1$ has cycle type $(2, 2, \ldots, 2)$, and $\tau_0$ and $\tau_1$ generate a transitive subgroup of the symmetric group. We think of these permutations as acting on the half-edges of the ribbon graph, where $\tau_0$ rotates half-edges anticlockwise around their adjacent vertex and $\tau_1$ swaps half-edges belonging to the same underlying edge. More generally, one can consider an {\em $m$-hypermap} as a pair of permutations $(\tau_0, \tau_1)$ such that $\tau_1$ has cycle type $(m, m, \ldots, m)$, and $\tau_0$ and $\tau_1$ generate a transitive subgroup of the symmetric group. For further information on these topics, one may consult the book of Lando and Zvonkin~\cite{lan-zvo}.

\begin{definition}
Define the {\em ribbon graph number} $A_{g,n}(d_1, d_2, \ldots, d_n)$ to be the weighted count of ribbon graphs of genus $g$ with $n$ labelled faces of degrees $d_1, d_2, \ldots, d_n$. The weight of a ribbon graph~$\Gamma$ is $\frac{1}{|\text{Aut}~\Gamma|}$, where $\text{Aut}~\Gamma$ denotes the group of face-preserving automorphisms. The corresponding 1-point invariant is denoted $a_g(d) = 2d \, A_{g,1}(2d)$. (The factor of $2d$ in this definition provides agreement with the work of Harer and Zagier~\cite{har-zag} and produces a simpler 1-point recursion.)

We analogously define $A_{g,n}^{m}(d_1, d_2, \ldots, d_n)$ to be the weighted count of $m$-hypermaps of genus $g$ with $n$ labelled faces of degrees $d_1, d_2, \ldots, d_n$. The corresponding 1-point invariant is denoted $a^{m}_g(d) = md \, A^{m}_{g,1}(md)$.
\end{definition}

The following result is a consequence of the work of Alexandrov, Lewanski and Shadrin, in which they show an equivalence between counting hypermaps and the notion of strictly monotone orbifold Hurwitz numbers~\cite{ale-lew-sha}.

\begin{lemma} \label{lem:ribbongraphs}
The ribbon graph numbers arise from taking $\mathbf{q} = (0, 1, 0, 0, \ldots)$ and $G(z) = 1+z$ in \cref{eq:partitionfunction}. In other words, we have
\begin{align*}
Z(\mathbf{p}; \mathbf{q}; \h) &= \sum_{\lambda \in {\mathcal P}} s_\lambda(p_1, p_2, \ldots) \, s_\lambda(0, \tfrac{1}{\h}, 0, \ldots) \, \prod_{\Box \in \lambda} (1 + c(\Box) \h) \\
&= \exp \bigg[ \sum_{g=0}^\infty \sum_{n=1}^\infty \sum_{d_1, d_2, \ldots, d_n = 1}^\infty A_{g,n}(d_1, d_2, \ldots, d_n) \, \frac{\h^{2g-2+n}}{n!} \, p_{d_1} p_{d_2} \cdots p_{d_n} \bigg].
\end{align*}
More generally, the enumeration of $m$-hypermaps arises from keeping $G(z) = 1+z$, but taking $q_m = 1$ and $q_i = 0$ for $i \neq m$.
\end{lemma}

Ribbon graphs can be considered as special cases of the more general notion of dessins d'enfant.

\begin{definition}
A {\em dessin d'enfant} is a ribbon graph whose vertices are coloured black and white such that every edge is adjacent to one vertex of each colour. An {\em isomorphism} between two dessins d'enfant is an isomorphism between their underlying ribbon graphs that preserves the vertex colouring.
\end{definition}

One obtains the notion of a ribbon graph by considering dessins d'enfant in which every black vertex has degree two. In that case, one can simply remove the degree two vertex and amalgamate the adjacent two edges into a single edge, to obtain a ribbon graph. Similarly, dessins d'enfant in which every black vertex has degree $m$ give rise to $m$-hypermaps.

\begin{definition}
Define the {\em dessin d'enfant number} $B_{g,n}(d_1, d_2, \ldots, d_n)$ to be the weighted count of dessins d'enfant of genus $g$ with $n$ labelled faces of degrees $2d_1, 2d_2, \ldots, 2d_n$. The weight of a dessin d'enfant~$\Gamma$ is $\frac{1}{|\text{Aut}~\Gamma|}$, where $\text{Aut}~\Gamma$ denotes the group of face-preserving automorphisms. The corresponding 1-point invariant is denoted $b_g(d) = d \, B_{g,1}(d)$.
\end{definition}

More generally, we can refine the enumeration by weighting with parameters that record the degrees of the black vertices.

\begin{definition}
Define the {\em double dessin d'enfant number} $\overline{B}_{g,n}(d_1, d_2, \ldots, d_n)$ to be the analogous weighted count of dessins d'enfant, where the weight of a dessin d'enfant $\Gamma$ with black vertices of degrees $\lambda_1, \lambda_2, \ldots, \lambda_\ell$ is $\frac{q_{\lambda_1} q_{\lambda_2} \cdots q_{\lambda_\ell}}{|\text{Aut}~\Gamma|}$. The corresponding 1-point invariant is denoted $\overline{b}_g(d) = d \, \overline{B}_{g,1}(d)$.
\end{definition}

In the above definition, $q_1, q_2, q_3, \ldots$ are indeterminates, so we have $\overline{B}_{g,n}(d_1, d_2, \ldots, d_n) \in \mathbb{Q}[q_1, q_2, q_3, \ldots]$. The following result generalises \cref{lem:ribbongraphs}.

\begin{lemma}
The double dessin d'enfant numbers arise from taking $\mathbf{q} = (q_1, q_2, q_3, \ldots)$ and $G(z) = 1+z$ in \cref{eq:partitionfunction}. In other words, we have
\begin{align*}
Z(\mathbf{p}; \mathbf{q}; \h) &= \sum_{\lambda \in {\mathcal P}} s_\lambda(p_1, p_2, \ldots) \, s_\lambda(\tfrac{q_1}{\h}, \tfrac{q_2}{\h}, \ldots) \, \prod_{\Box \in \lambda} (1 + c(\Box) \h) \\
&= \exp \bigg[ \sum_{g=0}^\infty \sum_{n=1}^\infty \sum_{d_1, d_2, \ldots, d_n = 1}^\infty \overline{B}_{g,n}(d_1, d_2, \ldots, d_n) \, \frac{\h^{2g-2+n}}{n!} \, p_{d_1} p_{d_2} \cdots p_{d_n} \bigg].
\end{align*}
\end{lemma}

One obtains the usual dessin d'enfant enumeration by setting $\mathbf{q} = (1, 1, 1, \ldots)$ in the double dessin d'enfant enumeration.
\begin{align*}
Z(\mathbf{p}; \mathbf{q}; \h) &= \exp \bigg[ \sum_{g=0}^\infty \sum_{n=1}^\infty \sum_{d_1, d_2, \ldots, d_n = 1}^\infty B_{g,n}(d_1, d_2, \ldots, d_n) \, \frac{\h^{2g-2+n}}{n!} \, p_{d_1} p_{d_2} \cdots p_{d_n} \bigg] \\
&= \sum_{\lambda \in {\mathcal P}} s_\lambda(p_1, p_2, \ldots) \, s_\lambda(\tfrac{1}{\h}, \tfrac{1}{\h}, \tfrac{1}{\h}, \ldots) \, \prod_{\Box \in \lambda} (1 + c(\Box) \h) \\
&= \sum_{\lambda \in {\mathcal P}} s_\lambda(p_1, p_2, \ldots) \, s_\lambda(\tfrac{1}{\h}, 0, 0, \ldots) \, \prod_{\Box \in \lambda} (1 + c(\Box) \h)^2
\end{align*}
The second equality here relies on the fact that $s_\lambda(\tfrac{1}{\h}, \tfrac{1}{\h}, \tfrac{1}{\h}, \ldots) = s_\lambda(\tfrac{1}{\h}, 0, 0, \ldots) \prod (1 + c(\Box) \h)$, which is a direct corollary of the hook-length and the hook-content formulas --- see \cref{eq:hookcontent}.

\begin{figure}[ht!]
\begin{center}
\begin{tabularx}{\textwidth}{cccX}
\rowcolor{lightblue} $d$ & $g$ & $a_g(d)$ & $\overline{b}_g(d)$ \\
\rowcolor{lightred} 1 & 0 & $1$ & $q_1$ \\
\rowcolor{lightblue} 2 & 0 & $2$ & $q_2 + q_1^2$ \\
\rowcolor{lightblue} 2 & 1 & $1$ & $0$ \\ 
\rowcolor{lightred} 3 & 0 & $5$ & $q_3 + 3 q_2 q_1 + q_1^3$ \\
\rowcolor{lightred} 3 & 1 & $10$ & $q_3$ \\ 
\rowcolor{lightblue} 4 & 0 & $14$ & $q_4 + 4q_3q_1 + 2q_2^2 + 6q_2q_1^2 + q_1^4$ \\
\rowcolor{lightblue} 4 & 1 & $70$ & $5q_4 + 4q_3q_1 + q_2^2$ \\
\rowcolor{lightblue} 4 & 2 & $21$ & $0$ \\ 
\rowcolor{lightred} 5 & 0 & $42$ & $q_5 + 5q_4q_1 + 5q_3q_2 + 10q_3q_1^2 + 10q_2^2q_1 + 10q_2q_1^3 + q_1^5$ \\
\rowcolor{lightred} 5 & 1 & $420$ & $15q_5 + 25q_4q_1 + 15q_3q_2 + 10q_3q_1^2 + 5q_2^2q_1$ \\
\rowcolor{lightred} 5 & 2 & $483$ & $8q_5$ \\ 
\rowcolor{lightblue} 6 & 0 & $132$ & $q_6 + 6q_5q_1 + 6q_4q_2 + 15q_4q_1^2 + 3q_3^2 + 30q_3q_2q_1 + 20q_3q_1^3 + 5q_2^3 + 30q_2^2q_1^2 + 15q_2q_1^4 + q_1^6$ \\
\rowcolor{lightblue} 6 & 1 & $2310$ & $35q_6 + 90q_5q_1 + 60q_4q_2 +75q_4q_1^2 +25q_3^2 + 90q_3q_2q_1 + 20q_3q_1^3 + 10q_2^3 + 15q_2^2q_1^2$ \\
\rowcolor{lightblue} 6 & 2 & $6468$ & $84 q_6 + 48q_5q_1 + 24q_4q_2 + 12q_3^2$ \\
\rowcolor{lightblue} 6 & 3 & $1485$ & $0$
\end{tabularx}
\end{center}
\caption{Table of ribbon graph numbers and double dessin d'enfant numbers.}
\end{figure}

\subsection{Bousquet-M\'{e}lou--Schaeffer numbers}

One can encode a dessin d'enfant via a pair $(\sigma_1, \sigma_2)$ of permutations acting on the edges. Here, $\sigma_1$ acts by rotating each edge anticlockwise around its adjacent black vertex and $\sigma_2$ acts by rotating each edge anticlockwise around its adjacent white vertex. The connectedness of the dessin d'enfant is encoded in the fact that the two permutations generate a transitive subgroup of the symmetric group. For more details, one can consult the extensive literature on dessins d'enfant~\cite{lan-zvo}. More generally, one has the notion of Bousquet-M\'{e}lou--Schaeffer numbers~\cite{bou-sch}.

\begin{definition}
For $m$ a positive integer, the {\em Bousquet-M\'{e}lou--Schaeffer (BMS) number} $B^{m}_{g,n}(d_1, d_2, \ldots, d_n)$ is equal to $\frac{1}{|\mathbf{d}|!}$ multiplied by the number of tuples $(\sigma_1, \sigma_2, \ldots, \sigma_m)$ of permutations in $S_{|\mathbf{d}|}$ such that
\begin{itemize}
\item $\sum_{i=1}^m (|\mathbf{d}|-k(\sigma_i)) = 2g - 2 + n + |\mathbf{d}|$, where $k(\sigma)$ denotes the number of cycles in $\sigma$;
\item $\sigma_1 \circ \sigma_2 \circ \cdots \circ \sigma_m$ has $n$ labelled cycles with lengths $d_1, d_2, \ldots, d_n$; and
\item $\sigma_1, \sigma_2, \ldots, \sigma_m$ generate a transitive subgroup of the symmetric group.
\end{itemize}
The corresponding 1-point invariant is denoted $b^m_g(d) = d \, B^m_{g,1}(d)$.
\end{definition}

\begin{lemma}
The $m$-BMS numbers arise from taking $\mathbf{q} = (1, 0, 0, \ldots)$ and $G(z) = (1+z)^m$ in \cref{eq:partitionfunction}. In other words, we have
\begin{align*}
Z(\mathbf{p}; \mathbf{q}; \h) &= \sum_{\lambda \in {\mathcal P}} s_\lambda(p_1, p_2, \ldots) \, s_\lambda(\tfrac{1}{\h}, 0, 0, \ldots) \, \prod_{\Box \in \lambda} (1 + c(\Box) \h)^m \\
&= \exp \bigg[ \sum_{g=0}^\infty \sum_{n=1}^\infty \sum_{d_1, d_2, \ldots, d_n = 1}^\infty B_{g,n}^{m}(d_1, d_2, \ldots, d_n) \, \frac{\h^{2g-2+n}}{n!} \, p_{d_1} p_{d_2} \cdots p_{d_n} \bigg].
\end{align*}
\end{lemma}

By the Riemann existence theorem, one can equivalently consider $B_{g,n}^{m}(d_1, d_2, \ldots, d_n)$ to be the weighted count of connected genus $g$ branched covers $f: (C; p_1, p_2, \ldots, p_n) \to (\mathbb{CP}^1; \infty)$ such that
\begin{itemize}
\item $f^{-1}(\infty) = d_1 p_1 + d_2 p_2 + \cdots + d_n p_n$;
\item all other ramification occurs at the $m$th roots of unity.
\end{itemize}
The weight of a branched cover $f: C \to \mathbb{CP}^1$ is $\frac{1}{|\text{Aut}~f|}$, where an automorphism of $f$ is a Riemann surface automorphism $\phi: C \to C$ such that $f \circ \phi = f$.

More generally, we can refine the enumeration by weighting by parameters that record the ramification profile at one of the roots of unity.

\begin{definition}
The {\em double Bousquet-M\'{e}lou--Schaeffer number} $\overline{B}^{m}_{g,n}(d_1, d_2, \ldots, d_n)$ is the weighted count of genus $g$ connected branched covers $f: (C; p_1, p_2, \ldots, p_n) \to (\mathbb{CP}^1; \infty)$ such that
\begin{itemize}
\item $f^{-1}(\infty) = d_1 p_1 + d_2 p_2 + \cdots + d_n p_n$;
\item all other ramification occurs at the $m$th roots of unity.
\end{itemize}
The weight of a branched cover with ramification profile $(\lambda_1, \lambda_2, \ldots, \lambda_\ell)$ over $\exp(\frac{2\pi i}{m})$ is $\frac{q_{\lambda_1} q_{\lambda_2} \cdots q_{\lambda_{\ell}}}{|\text{Aut}~f|}$. The corresponding 1-point invariant is denoted $\overline{b}^m_g(d) = d \, \overline{B}^m_{g,1}(d)$.
\end{definition}

These numbers arise from taking $\mathbf{q} = (q_1, q_2, q_3, \ldots)$ and $G(z) = (1+z)^{m-1}$ in \cref{eq:partitionfunction}.

\begin{figure}[ht!]
\begin{center}
\begin{tabularx}{\textwidth}{ccX}
\rowcolor{lightblue} $d$ & $g$ & $\overline{b}^3_g(d)$ \\ 
\rowcolor{lightred} 1 & 0 & $q_1$ \\ 
\rowcolor{lightblue} 2 & 0 & $q_2 + 2q_1^2$ \\ 
\rowcolor{lightblue} 2 & 1 & $q_2$ \\ 
\rowcolor{lightred} 3 & 0 & $q_3 + 6q_2q_1 + 5q_1^3$ \\ 
\rowcolor{lightred} 3 & 1 & $8q_3 + 12q_2q_1 + q_1^3$ \\ 
\rowcolor{lightred} 3 & 2 & $3q_3$ \\ 
\rowcolor{lightblue} 4 & 0 & $q_4 + 8q_3q_1 + 4q_2^2 + 28q_2q_1^2 + 14q_1^4$ \\ 
\rowcolor{lightblue} 4 & 1 & $30q_4 + 96q_3q_1 + 34q_2^2 + 100q_2q_1^2 + 10q_1^4$ \\ 
\rowcolor{lightblue} 4 & 2 & $93q_4 + 88q_3q_1 + 34q_2^2 + 16q_2q_1^2$ \\ 
\rowcolor{lightblue} 4 & 3 & $20q_4$ \\ 
\rowcolor{lightred} 5 & 0 & $q_5 + 10q_4q_1 + 10q_3q_2 + 45q_3q_1^2 + 45q_2^2q_1 + 120q_2q_1^3 + 42q_1^5$ \\ 
\rowcolor{lightred} 5 & 1 & $80q_5 + 400q_4q_1 + 280q_3q_2 + 770q_3q_1^2 + 560q_2^2q_1 + 700q_2q_1^3 + 70q_1^5$ \\ 
\rowcolor{lightred} 5 & 2 & $901q_5 + 1990q_4q_1 + 1290q_3q_2 + 1405q_3q_1^2 + 1055q_2^2q_1 + 380q_2q_1^3 + 8q_1^5$ \\ 
\rowcolor{lightred} 5 & 3 & $1650q_5 + 1200q_4q_1 + 820q_3q_2 + 180q_3q_1^2 + 140q_2^2q_1$ \\ 
\rowcolor{lightred} 5 & 4 & $248q_5$
\end{tabularx}
\end{center}
\caption{Table of double BMS-3 numbers.}
\end{figure}

\subsection{Hurwitz numbers}

Hurwitz numbers enumerate branched covers of the Riemann sphere. They were first studied by Hurwitz~\cite{hur} in the late nineteenth century, although interest in Hurwitz numbers has been revived in recent decades due to connections to enumerative geometry~\cite{eke-lan-sha-vai, oko-pan}, integrability~\cite{oko}, and topological recursion~\cite{bou-mar, eyn-mul-saf}.

\begin{definition}
The {\em simple Hurwitz number} $H_{g,n}(d_1, d_2, \ldots, d_n)$ is the weighted count of genus $g$ connected branched covers $f: (C; p_1, p_2, \ldots, p_n) \to (\mathbb{CP}^1; \infty)$ such that
\begin{itemize}
\item $f^{-1}(\infty) = d_1 p_1 + d_2 p_2 + \cdots + d_n p_n$; and
\item the only other ramification is simple and occurs at the $m$th roots of unity.
\end{itemize}
The weight of a branched cover $f$ is $\frac{1}{m! \, |\text{Aut}~f|}$, where we have $m = 2g-2+n+|\mathbf{d}|$ from the Riemann--Hurwitz formula. The corresponding 1-point invariant is denoted $h_g(d) = d \, H_{g,1}(d)$.
\end{definition}

Again, the Riemann existence theorem allows one to encode a branched cover via its monodromy representation, which makes connection with permutation factorisations. The result is the following algebraic description of simple Hurwitz numbers.

\begin{proposition}
The simple Hurwitz number $H_{g,n}(d_1, d_2, \ldots, d_n)$ is $\frac{1}{m! \, |\mathbf{d}|!}$ multiplied by the number of tuples $(\tau_1, \tau_2, \ldots, \tau_m)$ of transpositions in $S_{|\mathbf{d}|}$ such that
\begin{itemize}
\item $m = 2g-2+n+|\mathbf{d}|$;
\item $\tau_1 \circ \tau_2 \circ \cdots \circ \tau_m$ has labelled cycles of lengths $d_1, d_2, \ldots, d_n$; and
\item $\tau_1, \tau_2, \ldots, \tau_m $ generate a transitive subgroup of the symmetric group.
\end{itemize}
\end{proposition}

This algebraic description of simple Hurwitz numbers then leads naturally to the following result~\cite{oko}.

\begin{lemma}
The simple Hurwitz numbers arise from taking $\mathbf{q} = (1, 0, 0, \ldots)$ and $G(z) = \exp(z)$ in \cref{eq:partitionfunction}. In other words, we have
\begin{align*}
Z(\mathbf{p}; \mathbf{q}; \h) &= \sum_{\lambda \in {\mathcal P}} s_\lambda(p_1, p_2, \ldots) \, s_\lambda(\tfrac{1}{\h}, 0, 0, \ldots) \, \prod_{\Box \in \lambda} \exp(c(\Box) \h) \\
&= \exp \bigg[ \sum_{g=0}^\infty \sum_{n=1}^\infty \sum_{d_1, d_2, \ldots, d_n = 1}^\infty H_{g,n}(d_1, d_2, \ldots, d_n) \, \frac{\h^{2g-2+n}}{n!} \, p_{d_1} p_{d_2} \cdots p_{d_n} \bigg].\end{align*}
\end{lemma}

As with the enumerations considered previously in this section, one can consider a generalisation of the Hurwitz enumeration to its ``double'' counterpart~\cite{do-kar18}.

\begin{definition}
The {\em double Hurwitz number} $\overline{H}_{g,n}(d_1, d_2, \ldots, d_n)$ is the weighted count of genus $g$ connected branched covers $f: (C; p_1, p_2, \ldots, p_n) \to (\mathbb{CP}^1; \infty)$ such that
\begin{itemize}
\item $f^{-1}(\infty) = d_1 p_1 + d_2 p_2 + \cdots + d_n p_n$; 
\item the ramification profile over 0 is arbitrary; and
\item the only other ramification is simple and occurs at the $m$th roots of unity.
\end{itemize}
The weight of a branched cover $f$ with ramification profile $(\lambda_1, \lambda_2, \ldots, \lambda_\ell)$ over 0 is $\frac{q_{\lambda_1} q_{\lambda_2} \cdots q_{\lambda_{\ell}}}{m! \, |\mathrm{Aut}~ f|}$.
\end{definition}

Again, we have a natural double Schur function expansion for double Hurwitz number partition function~\cite{oko}.

\begin{lemma}
The double Hurwitz numbers arise from taking $\mathbf{q} = (q_1, q_2, q_3, \ldots)$ and $G(z) = \exp(z)$ in \cref{eq:partitionfunction}. In other words, we have
\begin{align*}
Z(\mathbf{p}; \mathbf{q}; \h) &= \sum_{\lambda \in {\mathcal P}} s_\lambda(p_1, p_2, \ldots) \, s_\lambda(\tfrac{q_1}{\h}, \tfrac{q_2}{\h}, \ldots) \, \prod_{\Box \in \lambda} \exp(c(\Box) \h) \\
&= \exp \bigg[ \sum_{g=0}^\infty \sum_{n=1}^\infty \sum_{d_1, d_2, \ldots, d_n = 1}^\infty \overline{H}_{g,n}(d_1, d_2, \ldots, d_n) \, \frac{\h^{2g-2+n}}{n!} \, p_{d_1} p_{d_2} \cdots p_{d_n} \bigg].
\end{align*}
\end{lemma}

\begin{figure}[ht!]
\begin{center}
\begin{tabularx}{\textwidth}{cccX}
\rowcolor{lightblue} $d$ & $g$ & $h_g(d)$ & $\overline{h}_g(d)$ \\ 
\rowcolor{lightred} 1 & 0 & 1 & $q_1$ \\ 
\rowcolor{lightred} 1 & 1 & 0 & $0$ \\ 
\rowcolor{lightred} 1 & 2 & 0 & $0$ \\ 
\rowcolor{lightblue} 2 & 0 & 1 & $q_2+q_1^2$ \\
\rowcolor{lightblue} 2 & 1 & $\frac{1}{6}$ & $\frac{1}{2}q_2 + \frac{1}{6}q_1^2$ \\
\rowcolor{lightblue} 2 & 2 & $\frac{1}{120}$ & $\frac{1}{24}q_2 + \frac{1}{120}q_1^2$ \\ 
\rowcolor{lightred} 3 & 0 & $\frac{3}{2}$ & $q_3 + 3q_2q_1 + \frac{3}{2}q_1^3$ \\
\rowcolor{lightred} 3 & 1 & $\frac{9}{8}$ & $3q_3 + \frac{9}{2}q_2q_1 + \frac{9}{8}q_1^3$ \\
\rowcolor{lightred} 3 & 2 & $\frac{27}{80}$ & $\frac{9}{4}q_3 + \frac{81}{40}q_2q_1 + \frac{27}{80}q_1^3$ \\ 
\rowcolor{lightblue} 4 & 0 & $\frac{8}{3}$ & $q_4 + 4q_3q_1 + 2q_2^2 + 8q_2q_1^2 + \frac{8}{3}q_1^4$ \\
\rowcolor{lightblue} 4 & 1 & $\frac{16}{3}$ & $10q_4 + 24q_3q_1 + \frac{28}{3}q_2^2 + \frac{80}{3}q_2q_1^2 + \frac{16}{3}q_1^4$ \\
\rowcolor{lightblue} 4 & 2 & $\frac{208}{45}$ & $\frac{82}{3}q_4 + \frac{216}{5}q_3q_1 + \frac{244}{15}q_2^2 + \frac{1456}{45}q_2q_1^2 + \frac{208}{45}q_1^4$ \\ 
\rowcolor{lightred} 5 & 0 & $\frac{125}{24}$ & $q_5 + 5q_4q_1 + 5q_3q_2 + \frac{25}{2}q_3q_1^2 + \frac{25}{2}q_2^2q_1 + \frac{125}{6}q_2q_1^3 + \frac{125}{24}q_1^5$ \\
\rowcolor{lightred} 5 & 1 & $\frac{3125}{144}$ & $25q_5 + \frac{250}{3}q_4q_1 + \frac{125}{2}q_3q_2 + \frac{3125}{24}q_3q_1^2 + \frac{625}{6}q_2^2q_1 + \frac{3125}{24}q_2q_1^3 + \frac{3125}{144}q_1^5$ \\
\rowcolor{lightred} 5 & 2 & $\frac{15625}{384}$ & $\frac{2125}{12}q_5 + \frac{1250}{3}q_4q_1 + \frac{6875}{24}q_3q_2 + \frac{21875}{48}q_3q_1^2 + \frac{3125}{9}q_2^2q_1 + \frac{15625}{48}q_2q_1^3 + \frac{15625}{384}q_1^5$
\end{tabularx}
\end{center}
\caption{Table of simple Hurwitz numbers and double Hurwitz numbers.}
\end{figure}

\subsection{Monotone Hurwitz numbers}

Monotone Hurwitz numbers first appeared in a series of papers by Goulden, Guay-Paquet and Novak, in which they arose as coefficients in the large $N$ asymptotic expansion of the Harish-Chandra--Itzykson--Zuber matrix integral over the unitary group $U(N)$~\cite{gou-gua-nov13a, gou-gua-nov13b, gou-gua-nov14}. Their definition resembles that of Hurwitz numbers, but with a monotonicity constraint imposed on the transpositions. This monotonicity condition is rather natural from the standpoint of the Jucys--Murphy elements in the symmetric group algebra $\mathbb{C}[S_{|\mathbf{d}|}]$. Monotone Hurwitz numbers are known to obey several analogous properties to Hurwitz numbers. For instance, there is a polynomial structure theorem~\cite{gou-gua-nov13b}, they are governed by topological recursion~\cite{do-dye-mat}, there is a quantum curve~\cite{do-dye-mat}, and there is an ELSV-type formula~\cite{ale-lew-sha, do-kar}.

\begin{definition}
The {\em simple monotone Hurwitz number} $M_{g,n}(d_1, d_2, \ldots, d_n)$ is $\frac{1}{|\mathbf{d}|!}$ multiplied by the number of tuples $(\tau_1, \tau_2, \ldots, \tau_m)$ of transpositions in $S_{|\mathbf{d}|}$ such that
\begin{itemize}
\item $m = 2g-2+n+|\mathbf{d}|$;
\item $\tau_1 \circ \tau_2 \circ \cdots \circ \tau_m$ has labelled cycles of lengths $d_1, d_2, \ldots, d_n$;
\item $\tau_1, \tau_2, \ldots, \tau_m$ generate a transitive subgroup of the symmetric group; and
\item if $\tau_i = (a_i ~ b_i)$ with $a_i < b_i$, then $b_1 \leq b_2 \leq \cdots \leq b_m$.
\end{itemize}
The corresponding 1-point invariant is denoted $m_g(d) = d \, M_{g,1}(d)$.
\end{definition}

\begin{lemma}
The monotone Hurwitz numbers arise from taking $\mathbf{q} = (1, 0, 0, \ldots)$ and $G(z) = \frac{1}{1-z}$ in \cref{eq:partitionfunction}. In other words, we have
\begin{align*}
Z(\mathbf{p}; \mathbf{q}; \h) &= \sum_{\lambda \in {\mathcal P}} s_\lambda(p_1, p_2, \ldots) \, s_\lambda(\tfrac{1}{\h}, 0, 0, \ldots) \, \prod_{\Box \in \lambda} \frac{1}{1-c(\Box) \h} \\
&= \exp \bigg[ \sum_{g=0}^\infty \sum_{n=1}^\infty \sum_{d_1, d_2, \ldots, d_n = 1}^\infty M_{g,n}(d_1, d_2, \ldots, d_n) \, \frac{\h^{2g-2+n}}{n!} \, p_{d_1} p_{d_2} \cdots p_{d_n} \bigg].
\end{align*}
\end{lemma}

Again, one can consider a generalisation of the monotone Hurwitz enumeration to its ``double'' counterpart.

\begin{definition}
The {\em double monotone Hurwitz number} $\overline{M}_{g,n}(d_1, d_2, \ldots, d_n)$ is the weighted count of tuples $(\sigma, \tau_1, \tau_2, \ldots, \tau_m)$ of transpositions in $S_{|\mathbf{d}|}$ such that
\begin{itemize}
\item $m = 2g-2+n+k(\sigma)$, where $k(\sigma)$ denotes the number of cycles in $\sigma$;
\item $\sigma \circ \tau_1 \circ \tau_2 \circ \cdots \circ \tau_m$ has labelled cycles of lengths $d_1, d_2, \ldots, d_n$;
\item $\sigma, \tau_1, \tau_2, \ldots, \tau_m$ generate a transitive subgroup of the symmetric group; and
\item if $\tau_i = (a_i ~ b_i)$ with $a_i < b_i$, then $b_1 \leq b_2 \leq \cdots \leq b_m$.
\end{itemize}
The weight of such a tuple with $\sigma$ of cycle type $(\lambda_1, \lambda_2, \ldots, \lambda_\ell)$ is $\frac{1}{|\mathbf{d}|!} \, q_{\lambda_1} q_{\lambda_2} \cdots q_{\lambda_{\ell}}$. The corresponding 1-point invariant is denoted $\overline{m}_g(d) = d \, \overline{M}_{g,1}(d)$.
\end{definition}

\begin{lemma}
The double monotone Hurwitz numbers arise from taking $\mathbf{q} = (q_1, q_2, q_3, \ldots)$ and $G(z) = \frac{1}{1-z}$ in \cref{eq:partitionfunction}. In other words, we have
\begin{align*}
Z(\mathbf{p}; \mathbf{q}; \h) &= \sum_{\lambda \in {\mathcal P}} s_\lambda(p_1, p_2, \ldots) \, s_\lambda(\tfrac{q_1}{\h}, \tfrac{q_2}{\h}, \ldots) \, \prod_{\Box \in \lambda} \frac{1}{1-c(\Box) \h} \\
&= \exp \bigg[ \sum_{g=0}^\infty \sum_{n=1}^\infty \sum_{d_1, d_2, \ldots, d_n = 1}^\infty \overline{M}_{g,n}(d_1, d_2, \ldots, d_n) \, \frac{\h^{2g-2+n}}{n!} \, p_{d_1} p_{d_2} \cdots p_{d_n} \bigg].
\end{align*}
\end{lemma}

\begin{figure}[ht!]
\begin{center}
\begin{tabularx}{\textwidth}{cccX}
\rowcolor{lightblue} $d$ & $g$ & $m_g(d)$ & $\overline{m}_g(d)$ \\ 
\rowcolor{lightred} 1 & 0 & $1$ & $q_1$ \\ 
\rowcolor{lightred} 1 & 1 & $1$ & $0$ \\ 
\rowcolor{lightred} 1 & 2 & $1$ & $0$ \\ 
\rowcolor{lightblue} 2 & 0 & $1$ & $q_2+q_1^2$ \\
\rowcolor{lightblue} 2 & 1 & $1$ & $q_2 + q_1^2$ \\
\rowcolor{lightblue} 2 & 2 & $1$ & $q_2 + q_1^2$ \\ 
\rowcolor{lightred} 3 & 0 & $2$ & $q_3 + 3q_2q_1 + 2q_1^3$ \\
\rowcolor{lightred} 3 & 1 & $10$ & $5q_3 + 15q_2q_1 + 10q_1^3$ \\
\rowcolor{lightred} 3 & 2 & $42$ & $21q_3 + 63q_2q_1 + 42q_1^3$ \\ 
\rowcolor{lightblue} 4 & 0 & $5$ & $q_4 + 4q_3q_1 + 2q_2^2 + 10q_2q_1^2 + 5q_1^4$ \\
\rowcolor{lightblue} 4 & 1 & $70$ & $15q_4 + 60q_3q_1 + 25q_2^2 + 140q_2q_1^2 + 70q_1^4$ \\
\rowcolor{lightblue} 4 & 2 & $735$ & $161q_4 + 644q_3q_1 + 252q_2^2 + 1470q_2q_1^2 + 735q_1^4$ \\ 
\rowcolor{lightred} 5 & 0 & $14$ & $q_5 + 5q_4q_1 + 5q_3q_2 + 15q_3q_1^2 + 15q_2^2q_1 + 35q_2q_1^3 + 14q_1^5$ \\
\rowcolor{lightred} 5 & 1 & $420$ & $35q_5 + 175q_4q_1 + 140q_3q_2 + 490q_3q_1^2 + 420q_2^2q_1 + 1050q_2q_1^3 + 420q_1^5$ \\
\rowcolor{lightred} 5 & 2 & $8778$ & $777q_5 + 3885q_4q_1 + 2835q_3q_2 + 10605q_3q_1^2 + 8505q_2^2q_1 + 21945q_2q_1^3 + 8778q_1^5$
\end{tabularx}
\end{center}
\caption{Table of monotone Hurwitz numbers and double monotone Hurwitz numbers.}
\end{figure}

\section{Double Schur function expansions} \label{sec:schur}

\subsection{Partition functions and 1-point invariants}

In the previous section, we established that for various choices of the formal power series $G(z)$ and the parameters $q_1, q_2, q_3, \ldots$, certain enumerative problems of geometric interest are stored in the partition function via the following equation.
\begin{align*}
Z(\mathbf{p}; \mathbf{q}; \h) &= \sum_{\lambda \in {\mathcal P}} s_\lambda(p_1, p_2, \ldots) \, s_\lambda(\tfrac{q_1}{\h}, \tfrac{q_2}{\h}, \ldots) \, \prod_{\Box \in \lambda} G(c(\Box) \h) \\
&= \exp \bigg[ \sum_{g=0}^\infty \sum_{n=1}^\infty \sum_{d_1, d_2, \ldots, d_n = 1}^\infty N_{g,n}(d_1, d_2, \ldots, d_n) \, \frac{\h^{2g-2+n}}{n!} \, p_{d_1} p_{d_2} \cdots p_{d_n} \bigg]
\end{align*}
The numbers $N_{g,n}(d_1, d_2, \ldots, d_n)$ have been referred to in the literature as {\em weighted Hurwitz numbers} and are known to enumerate certain paths in the Cayley graph of $S_{|\mathbf{d}|}$ generated by transpositions~\cite{ale-cha-eyn-har}. Furthermore, the partition function $Z(\mathbf{p}; \mathbf{q}; \h)$ is a hypergeometric tau-function for the Toda integrable hierarchy~\cite{orl-shc}.

We consider in particular the 1-point invariants $n_g(d) = d \, N_{g,1}(d)$ stored in the partition function.\footnote{The extra factor of $d$ in the definition of $n_g(d)$ will have little bearing on our results, but is introduced here for consistency with the original Harer--Zagier recursion and other results in the literature. We remark that the 1-point recursions are generally simpler with this normalisation, as can be witnessed from \cref{eq:harerzagier,eq:donorbury}.} In order to obtain information about these numbers, we deform the partition function via a parameter $s$ that keeps track of the unweighted degree in $p_1, p_2, p_3, \ldots$ and then extract the 1-point invariants by differentiation.
\begin{align*}
\left[ \frac{\partial}{\partial s} Z(s\mathbf{p}; \mathbf{q}; \h) \right]_{s=0} &= \sum_{\lambda \in {\mathcal P}} \left[ \frac{\partial}{\partial s} s_\lambda(sp_1, sp_2, \ldots) \right]_{s=0} \, s_\lambda(\tfrac{q_1}{\h}, \tfrac{q_2}{\h}, \ldots) \, \prod_{\Box \in \lambda} G(c(\Box) \h) \\
&= \sum_{g=0}^\infty \sum_{d = 1}^\infty N_{g,1}(d) \, \h^{2g-1} \, p_d
\end{align*}

At this stage, it is natural to introduce the so-called {\em principal specialisation} $p_d = x^d$ to record the degree via the single variable $x$. 
\begin{align} \label{eq:1pointpart}
\left[ \frac{\partial}{\partial s} Z(sx, sx^2, sx^3, \ldots; \mathbf{q}; \h) \right]_{s=0} 
&= \sum_{\lambda \in {\mathcal P}} \left[ \frac{\partial}{\partial s} s_\lambda(sx, sx^2, sx^3, \ldots) \right]_{s=0} \, s_\lambda(\tfrac{q_1}{\h}, \tfrac{q_2}{\h}, \ldots) \, \prod_{\Box \in \lambda} G(c(\Box) \h) \notag \\
&= \sum_{g=0}^\infty \sum_{d = 1}^\infty N_{g,1}(d) \, \h^{2g-1} \, x^d
\end{align}

\subsection{Schur function evaluations}

In this section, we deduce some facts about Schur functions that will be required at a later stage. We begin with the crucial observation that the evaluation of the Schur function appearing in \cref{eq:1pointpart} is zero unless $\lambda$ is a hook partition. Here, and throughout the paper, a {\em hook partition} refers to a partition of the form $(k, 1^{d-k})$, where $1 \leq k \leq d$.

\begin{lemma} \label{lem:schurevaluation}
\[
\left[ \frac{\partial}{\partial s} s_\lambda(sx, sx^2, sx^3, \ldots) \right]_{s=0} =
\begin{cases}
(-1)^{d-k} \, \frac{x^d}{d}, & \text{if $\lambda = (k, 1^{d-k})$ is a hook partition,} \\
0, & \text{otherwise}.
\end{cases}
\]
\end{lemma}

\begin{proof}
The lemma follows from the hook-content formula~\cite{mac}, which states that
\begin{equation} \label{eq:hookcontent}
s_\lambda(s, s, s, \ldots) = \prod_{\Box \in \Lambda} \frac{s+c(\Box)}{h(\Box)},
\end{equation}
where $c(\Box)$ and $h(\Box)$ denote the content and hook-length of a box in the Young diagram for $\lambda$, respectively.

If $\lambda$ is a non-empty partition that is not a hook, then its Young diagram contains at least two boxes with content 0. So the hook-content formula implies that $s_\lambda(s, s, s, \ldots)$ is a polynomial divisible by $s^2$ and it follows that 
\[
\left[ \frac{\partial}{\partial s} s_\lambda(sx, sx^2, sx^3, \ldots) \right]_{s=0} = 0.
\]

If $\lambda = (k, 1^{d-k})$ is a hook partition, then its hook-lengths are $\{1, 2, \ldots, k-1\} \cup \{1, 2 ,\ldots, d-k\} \cup \{d\}$, while its contents are $\{1, 2, \ldots, k-1\} \cup \{-1, -2, \ldots, -(d-k)\} \cup \{0\}$. Thus, we obtain
\[
s_\lambda(s, s, s, \ldots) = (-1)^{d-k} \, \frac{(s+k-1) (s+k-2) \cdots (s+k-d)}{d (k-1)! (d-k)!}.
\]
By directly differentiating with respect to $s$ and evaluating at $s = 0$, we obtain
\[
\left[ \frac{\partial}{\partial s} s_\lambda(s, s, \ldots) \right]_{s=0} = \frac{(-1)^{d-k}}{d}.
\]
The powers of $x$ appearing in the statement of the lemma can be reinstated, using the fact that Schur functions are weighted homogeneous.
\end{proof}

Now use \cref{lem:schurevaluation} in \cref{eq:1pointpart} to obtain the following.
\begin{align*}
\left[ \frac{\partial}{\partial s} Z(sx, sx^2, sx^3, \ldots; \mathbf{q}; \h) \right]_{s=0} &= \sum_{g=0}^\infty \sum_{d = 1}^\infty N_{g,1}(d) \h^{2g-1} x^d \\
&= \sum_{d=1}^\infty \sum_{k=1}^d (-1)^{d-k} \frac{x^d}{d} \, s_{(k, 1^{d-k})}(\tfrac{q_1}{\h}, \tfrac{q_2}{\h}, \ldots) \, \prod_{\Box \in \lambda} G(c(\Box) \h)
\end{align*}
Extracting the $x^d$ coefficient yields the following result.

\begin{lemma} \label{lem:1pointpart}
The 1-point invariants $n_g(d) = d \, N_{g,1}(d)$ defined by \cref{eq:partitionfunction} satisfy
\[
\sum_{g=0}^\infty n_g(d) \,\h^{2g-1} = \sum_{k=1}^d (-1)^{d-k} \, s_{(k, 1^{d-k})}(\tfrac{q_1}{\h}, \tfrac{q_2}{\h}, \ldots) \, \prod_{i=1}^d G((k-i)\h),
\]
for every positive integer $d$.
\end{lemma}

We will later be interested in setting the parameter $q_i = 0$ for $i$ sufficiently large. In this case, we write $s_\lambda(\tfrac{q_1}{\h}, \tfrac{q_2}{\h}, \ldots, \tfrac{q_r}{\h})$ to mean the Schur function $s_\lambda(\tfrac{q_1}{\h}, \tfrac{q_2}{\h}, \ldots)$ evaluated at $q_{r+1} = q_{r+2} = \cdots = 0$.

We complete the section by presenting the following relations concerning Schur functions, which will be useful for the next section~\cite{mac}.
\begin{lemma} \label{lem:hookschur}
The Schur function indexed by the hook $(k, 1^{d-k})$ can be expressed as
\[
s_{(k, 1^{d-k})}(\mathbf{p}) = \sum_{j=1}^k (-1)^{j+1} \, h_{k-j}(\mathbf{p}) \, e_{d-k+j}(\mathbf{p}).
\]
Here, $h_n$ and $e_n$ respectively denote the homogeneous and elementary symmetric functions, which can in turn be expressed in terms of power sum symmetric functions via
\[
\sum_{n=0}^\infty h_n(\mathbf{p}) \, x^n = \exp \bigg[ \sum_{k=1}^\infty \frac{p_k}{k} x^k \bigg] \qquad \text{and} \qquad \sum_{n=0}^\infty e_n(\mathbf{p}) \, x^n = \exp \bigg[ \sum_{k=1}^r (-1)^{k-1} \frac{p_k}{k} x^k \bigg].
\]
In the case $p_1 = s$ and $p_k = 0$ for $k \geq 2$, the above expression evaluates to
\[
s_{(k, 1^{d-k})}(s, 0, 0, \ldots) = \binom{d-1}{k-1} \frac{s^d}{d!}.
\]
\end{lemma}

\section{Recursions for 1-point functions} \label{sec:recursion}

\subsection{Holonomic sequences and functions}

A sequence $a_0, a_1, a_2, \ldots$ is said to be {\em holonomic over $\mathbb{K}$} if the terms satisfy a non-zero linear difference equation of the form
\begin{equation}
p_r(d) \, a_{d+r} + p_{r-1}(d) \, a_{d+r-1} + \cdots + p_1(d) \, a_{d+1} + p_0(d) \, a_d = 0,
\end{equation}
where $p_0, p_1, \ldots, p_r$ are polynomials over the field $\mathbb{K}$ of characteristic 0. Moreover, a formal power series $A(x) = \displaystyle\sum_{d=0}^\infty a_d \, x^d$ is said to be {\em holonomic over $\mathbb{K}$} if it satisfies a non-zero linear differential equation of the form
\begin{equation}
\left[ P_r(x) \, \frac{\partial^r}{\partial x^r} + P_{r-1}(x) \, \frac{\partial^{r-1}}{\partial x^{r-1}} + \cdots + P_1(x) \, \frac{\partial}{\partial x} + P_0(x) \right] A(x) = 0,
\end{equation}
where $P_0, P_1, \ldots, P_r$ are polynomials over $\mathbb{K}$. The dual use of the term ``holonomic'' is due to the elementary fact that the sequence $a_0, a_1, a_2, \ldots$ is holonomic over $\mathbb{K}$ if and only if the formal power series $a_0 + a_1 x + a_2 x^2 + \cdots$ is holonomic over $\mathbb{K}$. For our applications, we will use the ground field $\mathbb{K} = \mathbb{C}(\h)$.

\begin{lemma} \label{lem:holonomic}
A 1-point recursion exists for the numbers $n_g(d)$ in the sense of~\cref{def:1pointrecursion} if and only if the formal power series
\[
F(x, \h) = \sum_{d=1}^\infty \sum_{g=0}^\infty n_g(d) \, \h^{2g-1} \, x^d
\]
is holonomic over $\mathbb{C}(\h)$.
\end{lemma}

\begin{proof}
If $F(x, \h)$ is holonomic, then there exist polynomials $P_0, P_1, \ldots, P_r$ with coefficients in $\mathbb{C}(\h)$ such that
\[
\left[ P_r(x) \, \frac{\partial^r}{\partial x^r} + P_{r-1}(x) \, \frac{\partial^{r-1}}{\partial x^{r-1}} + \cdots + P_1(x) \, \frac{\partial}{\partial x} + P_0(x) \right] F(x, \h) = 0.
\]
One can assume that the coefficients of $P_0, P_1, \ldots, P_r$ actually lie in $\mathbb{C}[\h]$, by clearing denominators in the equation above. Thus, the equation takes the form 
\begin{equation} \label{eq:holonomiclemma}
\left[ \sum _{i,j,k=0}^{\text{finite}} C_{ijk} \, \h^i x^j \frac{\partial^k}{\partial x^k} \right] F(x, \h) = 0,
\end{equation}
for some complex constants $C_{ijk}$. Applying $C_{ijk} \, \h^i x^j \frac{\partial^k}{\partial x^k}$ to a term $n_g(d) \, \h^{2g-1} \, x^d$ in the expansion for $F(x, \h)$ has the effect of shifting the powers of $\h$ and $x$, and introducing a factor that is polynomial in~$d$. So after collecting terms in the resulting equation, one obtains a relation of the form of \cref{eq:1pointrecursion}. Therefore, there exists a 1-point recursion for the numbers $n_g(d)$.

Conversely, suppose that there exists a 1-point recursion for the numbers $n_g(d)$, so there exists a relation of the form of \cref{eq:1pointrecursion}. Multiplying both sides by $\h^{2g-1} \, x^d$ and summing over $g$ and $d$ yields
\[
\sum_{d=1}^\infty \sum_{g=0}^\infty \sum_{i=0}^{i_{\max}} \sum_{j=0}^{j_{\max}} p_{ij}(d) \, n_{g-i}(d-j) \, \h^{2g-1} \, x^d = 0.
\]
Now replace $p_{ij}(d) \, x^d$ with $p_{ij} \big(x \frac{\partial}{\partial x} \big) \, x^d$ and reindex the summations over $d$ and $g$ to obtain
\[
\sum_{d=1}^\infty \sum_{g=0}^\infty \sum_{i=0}^{i_{\max}} \sum_{j=0}^{j_{\max}} p_{ij}\big( x\tfrac{\partial}{\partial x} \big) \, \h^{2i} \, x^j \, n_g(d) \h^{2g-1} \, x^d = 0 \qquad \Rightarrow \qquad \bigg[ \sum_{i=0}^{i_{\max}} \sum_{j=0}^{j_{\max}} p_{ij}\big( x\tfrac{\partial}{\partial x} \big) \, \h^{2i} \, x^j \bigg] F(x, \h) = 0.
\]
This final equation can be expressed in the form of \cref{eq:holonomiclemma} by applying the commutation relation $[\frac{\partial}{\partial x}, x] = 1$. It then follows that $F(x, \h)$ is holonomic over $\mathbb{C}(\h)$.
\end{proof}

The following result lists some closure properties, which provide standard tools to prove holonomicity~\cite{kau-pau}.

\begin{proposition} \label{prop:holonomic}
Let $A(x) = \displaystyle\sum_{d=0}^\infty a_d \, x^d$ and $B(x) = \displaystyle\sum_{d=0}^\infty b_d \, x^d$ be holonomic over a field $\mathbb{K}$ of characteristic zero. Then
\begin{enumerate} [label={(\alph*)}, noitemsep]
\item $\alpha A(x) + \beta B(x)$ is holonomic for all $\alpha, \beta \in \mathbb{K}$;
\item the Cauchy product $A(x) \, B(x)$ and the Hadamard product $\big (a_n b_n \big)_{n = 0, 1, 2, \ldots}$ are holonomic;
\item the derivative $a'(x)$ and the forward shift $\big( a_{n+1} \big)_{n=0, 1, 2, \ldots}$ are holonomic; and
\item the integral $\int^x A(x)\, \dd x$ and the indefinite sum $\big( \sum_{k=0}^n a_k \big)_{n=0, 1, 2, \ldots}$ are holonomic.
\end{enumerate}
\end{proposition}

\begin{definition}
We define the {\em order} and {\em degree} of the difference equation 
\[
p_r(d) \, a_{d+r} + p_{r-1}(d) \, a_{d+r-1} + \cdots + p_1(d) \, a_{d+1} + p_0(d) \, a_d = 0
\]
to be $r$ (assuming $p_r(d) \neq 0$) and $\max \{\deg p_0, \deg p_1, \ldots, \deg p_r\}$, respectively. Similarly, we define the {\em order} and {\em degree} of the differential equation
\[
\left[ P_r(x) \, \frac{\partial^r}{\partial x^r} + P_{r-1}(x) \, \frac{\partial^{r-1}}{\partial x^{r-1}} + \cdots + P_1(x) \, \frac{\partial}{\partial x} + P_0(x) \right] A(x) = 0
\]
to be $r$ (assuming $P_r(x) \neq 0$) and $\max \{\deg P_0, \deg P_1, \ldots, \deg P_r\}$, respectively.
\end{definition}

\begin{remark*}
Note that for a fixed holonomic sequence or function, there are difference or differential operators of many possible orders and degrees that annihilate it. Furthermore, it is not generally true that there exists such an operator that simultaneously minimises both the order and the degree. Thus, one does not usually refer to the order and degree of a holonomic sequence or function itself, but to the order and degree of a particular operator.
\end{remark*}

\subsection{Multivariate holonomic functions}

There are competing ways in which the notion of holonomicity may be generalised to the case of many variables, but the following is well-suited to our purposes. Let $\mathbf{x} = (x_1, x_2, \ldots, x_n)$ and let $\mathbb{K}[[\mathbf{x}]] = \mathbb{K}[[x_1, x_2, \ldots, x_n]]$. A multivariate formal power series $A(\mathbf{x}) \in \mathbb{K}[[\mathbf{x}]]$ is said to be {\em holonomic over~$\mathbb{K}$} --- also commonly known as {\em D-finite} --- if the set of derivatives
\[
\left\{ \left. \frac{\partial^{i_1+i_2+\cdots+i_n}}{\partial x_1^{i_1} \, \partial x_2^{i_2} \, \cdots \, \partial x_n^{i_n}} A(\mathbf{x}) ~\right|~ i_1, i_2, \ldots, i_n \in \mathbb{Z}_{\geq 0} \right\}
\]
lie in a finite-dimensional vector space over $\mathbb{K}(\mathbf{x})$. This is equivalent to the fact that $A(\mathbf{x})$ satisfies a system of linear partial differential equations of the form 
\begin{equation}
\left[ P_{i, r}(\mathbf{x}) \, \frac{\partial^r}{\partial x_i^r} +P_{i, r-1}(\mathbf{x}) \, \frac{\partial^{r-1}}{x_i^{r-1}} + \cdots + P_{i, 0}(\mathbf{x}) \right] A(\mathbf{x}) = 0, \qquad \text{for } i=1, 2, \ldots, n,
\end{equation}
where $P_{ij}(x)\in \mathbb{K}[x]$. Clearly, the case $n=1$ recovers the definition of a holonomic function described earlier.

\begin{definition}
For $A(\mathbf{x}) = \displaystyle\sum_{i_1, i_2, \ldots, i_n} a(i_1, i_2, \ldots, i_n) \, x_1^{i_1} x_2^{i_2} \cdots x_n^{i_n} \in \mathbb{K}[[\mathbf{x}]]$ and integers $1 \leq k < \ell \leq n$, define the {\em primitive diagonal}
\[
I_{k\ell}(A(\mathbf{x})) = \sum_{i_1, \ldots, \widehat{i_\ell}, \ldots, i_n} a(i_1, \ldots, i_k, \ldots, i_k, \ldots, i_n) \, x_1^{i_1} \cdots x_k^{i_k} \cdots \widehat{x_{\ell}^{i_\ell}} \cdots x_n^{i_n},
\]
where the hats denote omission of the index $i_\ell$ and the term $x_\ell^{i_\ell}$.
\end{definition}

For example, taking $k=1$, $\ell=2$ and $n = 4$ leads to $I_{12}(A(x_1, x_2, x_3, x_4)) = \displaystyle\sum_{i_1, i_3, i_4} a(i_1,i_1,i_3, i_4) \, x_1^{i_1} x_3^{i_3} x_4^{i_4}$.

The following result lists some closure properties for multivariate holonomic functions~\cite{MR929767, MR587530}.

\begin{proposition} \label{diagtheorem} 
Let $A(\mathbf{x}) = \sum a(i_1, i_2, \ldots, i_n) \, x_1^{i_1} x_2^{i_2} \cdots x_n^{i_n}$ and $B(\mathbf{x}) = \sum b(i_1, i_2, \ldots, i_n) \, x_1^{i_1} x_2^{i_2} \cdots x_n^{i_n}$ be holonomic functions over a field $\mathbb{K}$ of characteristic zero. Then
\begin{enumerate} [label={(\alph*)}, noitemsep]
\item the primitive diagonal $I_{k\ell}(A(\mathbf{x}))$ is holonomic for all $1 \leq k < \ell \leq n$;
\item the Cauchy product $A(\mathbf{x}) \, B(\mathbf{x})$ is holonomic;
\item the Hadamard product $A(\mathbf{x}) * B(\mathbf{x}) = \sum a(i_1, i_2, \ldots, i_n) \, b(i_1, i_2, \ldots, i_n) \, x_1^{i_1} x_2^{i_2} \cdots x_n^{i_n}$ is holonomic; and
\item the formal power series
\[
\sum_{(i_1, i_2, \ldots, i_n) \in C} a(i_1, i_2, \ldots, i_n) \, x_1^{i_1} x_2^{i_2} \cdots x_n^{i_n}
\]
is holonomic if $C \subseteq \mathbb{Z}_{\geq 0}^n$ is defined by a finite set of inequalities of the form $\sum a_k i_k + b \geq 0$, where $a_1, a_2, \ldots, a_n, b \in \mathbb{Z}$.
\end{enumerate}
\end{proposition}

\subsection{Existence of 1-point recursions}

We begin by proving the existence of 1-point recursions in the {\em simple} case when $\mathbf{q} = (1, 0, 0, \ldots)$. (The word ``simple'' has been ported from the context of Hurwitz numbers to this more general setting.)

\begin{theorem} \label{thm:requals1}
Let $G(z) \in \mathbb{C}(z)$ be a rational function and let $\mathbf{q} = (1, 0, 0, \ldots)$. Define the numbers $n_g(d) = d \, N_{g,1}(d)$ via \cref{eq:partitionfunction}. Then the numbers $n_g(d)$ satisfy a 1-point recursion in the sense of~\cref{def:1pointrecursion}.
\end{theorem}

\begin{proof}
We define $n(d)$ and calculate it as follows.
\begin{align} \label{eq:1pointgf}
n(d) &= \sum_{g=0}^\infty n_g(d) \, \h^{2g-1} \nonumber \\
&= \sum_{k=1}^d (-1)^{d-k} \, s_{(k, 1^{d-k})}(\tfrac{1}{\h}, 0, 0, \ldots) \, \prod_{i=1}^d G((k-i)\h) & \text{(\cref{lem:1pointpart})} \nonumber \\
&= \frac{1}{d! \, \h^d}\sum_{k=1}^d (-1)^{d-k} \binom{d-1}{k-1}\prod_{i=1}^d G((k-i)\h) & \text{(\cref{lem:hookschur})} \nonumber \\
&= \frac{1}{d \, \h^d} \sum_{k=1}^d (-1)^{d-k} \frac{1}{(k-1)! \, (d-k)!}\prod_{i=1}^d G((k-i)\h)
\end{align}

Define the sequences
\[
u_k = \frac{1}{(k-1)! \, \h^k} \prod_{i=1}^k G((i-1)\h) \qquad \text{and} \qquad v_k = \frac{(-1)^k}{k! \, \h^k}\prod_{i=1}^{k} G(-i\h).
\]
These are holonomic over $\mathbb{C}(\h)$ since the ratios $\frac{u_{k+1}}{u_k} = \frac{G(k\h)}{k\h}$ and $\frac{v_{k+1}}{v_k} = -\frac{G(-(k+1)\h)}{(k+1) \h}$ are rational functions of $k$ with coefficients from $\mathbb{C}(\h)$. Hence, parts (b) and (c) of \cref{prop:holonomic} implies that the sequence
\[
n(d) = \frac{1}{d} \sum_{k=1}^d u_k \, v_{d-k}
\]
is holonomic over $\mathbb{C}(\h)$. So \cref{lem:holonomic} guarantees the existence of a 1-point recursion for $n_g(d)$.
\end{proof}

To tackle the case of general weights $\mathbf{q} = (q_1, q_2, \ldots, q_r, 0, 0, \ldots)$, we use the following lemma.

\begin{lemma} \label{lem:mainlem}
If $a_d, b_d, u_d, v_d$ are holonomic sequences, then so is
\[
s_d = \sum_{k=1}^d a_k b_{d-k}\sum_{\ell=0}^{k-1} u_{\ell} v_{d-\ell}.
\] 
\end{lemma}

\begin{proof}
Define the generating functions
\[
A(x_1) = \sum_{n=1}^\infty a_n x_1^n, \qquad B(x_2) = \sum_{n=0}^\infty b_n x_2^n, \qquad U(x_3, x_4) = \sum_{n=0}^\infty u_n (x_3x_4)^n, \qquad V(x_4) = \sum_{n=1}^\infty v_n x_4^n.
\]
Observe that each of these is a holonomic function in the appropriate variables. Since Cauchy products preserve holonomicity --- see part (b) of \cref{diagtheorem} --- we know that
\[
H(x_3, x_4) = \frac{x_3}{1-x_3} \, U(x_3, x_4) \, V(x_4) = \sum_{k=1}^\infty \sum_{n=1}^\infty \left( \sum_{\ell=0}^{k-1} u_\ell v_{n-\ell} \right) x_3^k x_4^n
\] 
is holonomic. (We interpret the inner summation by discarding any terms that involve $v_{n-\ell}$ with $n-\ell \leq 0$.) By part (d) of \cref{diagtheorem}, restricting to the terms satisfying $n-k\geq 0$, we obtain the holonomic function
\[
\widehat{H}(x_3, x_4) = \sum_{n \geq k \geq 1} \left( \sum_{\ell=0}^{k-1} u_\ell v_{n-\ell} \right) x_3^k x_4^n.
\]

Then 
\[
L(x_1, x_2, x_3, x_4) = A(x_1) \, B(x_2) \, \widehat{H}(x_3,x_4) = \sum_{i=1}^\infty \sum_{j=0}^\infty \sum_{n \geq k \geq 1} a_i b_j \left( \sum_{\ell=0}^{k-1} u_\ell v_{n-\ell} \right) x_1^i x_2^j x_3^k x_4^n
\]
is holonomic by closure under Cauchy products. Invoking part (a) of \cref{diagtheorem}, we know that
\[ 
I_{13}(L(x_1, x_2, x_3, x_4)) = \sum_{k=1}^\infty \sum_{j=0}^\infty \sum_{n=k}^\infty a_k b_j \left( \sum_{\ell=0}^{k-1} u_\ell v_{n-\ell} \right) x_1^k x_2^j x_4^n
\]
is holonomic. Now use part (d) of \cref{diagtheorem} with the inequalities $j+k-n \geq 0$ and $-j-k+n \geq 0$ --- in other words, restricting to $j = n-k$ --- to deduce holonomicity of
\[
\widehat{L}(x_1, x_2, x_4) = \sum_{k=1}^\infty \sum_{n=k}^\infty a_k b_{n-k} \left( \sum_{\ell=0}^{k-1} u_\ell v_{n-\ell} \right) x_1^k x_2^{n-k} x_4^n.
\]
By evaluating this formal power series at $x_1=1$, $x_2=1$ and $x_4=x$ --- which clearly preserves holonomicity --- we obtain the desired result.
\end{proof}

We are now in a position to prove \cref{thm:main}, which we restate in the following way.

\begin{theorem} \label{thm:main2}
Let $G(z) \in \mathbb{C}(z)$ be a rational function with $G(0) = 1$ and let $\mathbf{q} = (q_1, q_2, \ldots, q_r, 0, 0, \ldots)$. Define the numbers $n_g(d) = d \, N_{g,1}(d)$ via \cref{eq:partitionfunction}. Then the generating function
\begin{equation} \label{eq:mainthm}
\sum_{d=1}^\infty \sum_{g=0}^\infty n_g(d) \,\h^{2g-1} \, x^d
\end{equation}
is holonomic over $\mathbb{C}(\h)$, so the numbers $n_g(d)$ satisfy a 1-point recursion in the sense of~\cref{def:1pointrecursion}.
\end{theorem}

\begin{proof}
We calculate the coefficient $n(d)$ of $x^d$ in \cref{eq:mainthm} as follows.
\begin{align*}
n(d) &= \sum_{g=0}^\infty n_g(d) \, \h^{2g-1} \\
&= \sum_{k=1}^d (-1)^{d-k} \, s_{(k, 1^{d-k})}(\tfrac{q_1}{\h}, \tfrac{q_2}{\h}, \ldots, \tfrac{q_r}{\h}) \, \prod_{i=1}^d G((k-i)\h) & \text{(\cref{lem:1pointpart})} \\
 &= \sum_{k=1}^d (-1)^{d-k} \prod_{i=1}^d G((k-i)\h) \sum_{j=1}^k (-1)^{j+1} \, h_{k-j}(\tfrac{q_1}{\h}, \tfrac{q_2}{\h}, \ldots, \tfrac{q_r}{\h}) \, e_{d-k+j}(\tfrac{q_1}{\h}, \tfrac{q_2}{\h}, \ldots, \tfrac{q_r}{\h}) & \text{(\cref{lem:hookschur})} \\
 &= \sum_{k=1}^d \prod_{i=1}^{k} G((i-1)\h) \prod_{i=1}^{d-k} G(-i\h) \sum_{\ell=0}^{k-1} h_\ell(\tfrac{q_1}{\h}, \tfrac{q_2}{\h}, \ldots, \tfrac{q_r}{\h}) \, (-1)^{d-\ell+1} \, e_{d-\ell}(\tfrac{q_1}{\h}, \tfrac{q_2}{\h}, \ldots, \tfrac{q_r}{\h})
\end{align*}

Now define the sequences
\[
a_n = \prod_{i=1}^{n} G((i-1)\h), \quad b_n = \prod_{i=1}^{n} G(-i\h), \quad u_n = h_n(\tfrac{q_1}{\h}, \tfrac{q_2}{\h}, \ldots,\tfrac{q_{r}}{\h}), \quad v_n = (-1)^{n+1} \, e_n(\tfrac{q_1}{\h}, \tfrac{q_2}{\h}, \ldots,\tfrac{q_{r}}{\h}).
 \] 
The first two are holonomic over $\mathbb{C}(\h)$ since the ratios $\frac{a_{n+1}}{a_n} = G(n\h)$ and $\frac{b_{n+1}}{b_n} = G(-(n+1)\h)$ are rational functions of $n$ with coefficients from $\mathbb{C}(\h)$. The last two are holonomic over $\mathbb{C}(\h)$ due to \cref{lem:hookschur}, from which we deduce that
\[
\left[ \h\frac{\partial}{\partial x} - \sum_{k=1}^r q_k x^{k-1} \right] \left( \sum_{n=0}^\infty u_n x^n \right) = 0 \qquad \text{and} \qquad \left[ \h\frac{\partial}{\partial x} + \sum_{k=1}^r (-1)^k q_k x^{k-1} \right] \left( \sum_{n=0}^\infty v_n x^n \right) = 0.
\]

Hence, \cref{lem:mainlem} implies that the sequence
\[
n(d) = \sum_{k=1}^{d} a_kb_{d-k}\sum_{\ell=0}^{k-1}u_{\ell}v_{d-\ell}
\]
is holonomic over $\mathbb{C}(\h)$. It then follows from \cref{lem:holonomic} that there exists a 1-point recursion for the numbers $n_g(d)$.
\end{proof}

\subsection{Algorithms for 1-point recursions} \label{subsec:algorithm}

One of the features of the theory of holonomic sequences and functions is the fact that theoretical results can often be turned into effective algorithms. Although \cref{thm:main2} only asserts the existence of 1-point recursions, its proof can be converted into an algorithm to calculate them from the initial data of the rational function $G(z)$ and the positive integer $r$ that records the number of non-zero weights $\mathbf{q} = (q_1, q_2, \ldots, q_r)$. For example, a naive though feasible approach would be to express the putative 1-point recursion as
\[
\sum_{i=0}^D \sum_{j=0}^R a_{ij} \, d^i \, n(d-j) = 0,
\]
and treat this as a linear system in the $(D+1)(R+1)$ variables $a_{ij} \in \mathbb{C}(\h)$. One obtains a linear constraint for each positive integer $d$, so a finite number of these allows for the computation of the 1-point recursion.

In order to implement this approach, one requires explicit and simultaneous bounds on the degree $D$ and the order $R$ of such a recursion. We remark that it is possible to obtain such bounds in terms of the degree of $G(z)$ and the positive integer $r$. Begin with the operators that annihilate the generating functions for the sequences $a_n, b_n, u_n, v_n$ that appear in the proof of \cref{thm:main2}. Then use known bounds for the degree and order of operators that annihilate functions obtained by the holonomicity closure properties used in the proof --- namely, Cauchy product, taking diagonals, restricting summations, and evaluation. We do not pursue these calculations in the current work.

There are more efficient algorithms for computing with holonomic functions that are implemented in the \texttt{gfun} package for {\sc Maple}~\cite{sal-zim}. For example, we demonstrate below how the previously unknown 1-point recursion for monotone Hurwitz numbers may be derived from several lines of code.

\begin{example}
The proof of \cref{thm:requals1} implies that monotone Hurwitz numbers satisfy the relation
\[
m(d) = \sum_{g=0}^\infty m_g(d) \,\h^{2g-1} = \frac{1}{d} \sum_{k=1}^d u_k \, v_{d-k},
\]
where $\frac{u_{k+1}}{u_k} = \frac{G(k\h)}{k\h}$ and $\frac{v_{k+1}}{v_k} = -\frac{G(-(k+1)\h)}{(k+1) \h}$. So the sequence $m(d)$ can be obtained by taking the Cauchy product of $u_k$ and $v_k$, and then taking the Hadamard product of the result and the sequence $\frac{1}{k}$. The following shows several lines of hopefully self-explanatory {\sc Maple} code that produce a 1-point recursion for monotone Hurwitz numbers.

\MapleInput{with(gfun):} \\
\MapleInput{$G(z):=\frac{1}{1-z}$:} \\
\MapleInput{rec1:=\{d*hbar*m(d+1)-G(d*hbar)*m(d)=0, m(0)=0, m(1)=1\}:} \\
\MapleInput{rec2:=\{(d+1)*hbar*m(d+1)+G(-(d+1)*hbar)*m(d)=0, m(1)=-G(-hbar)\}:} \\
\MapleInput{rec3:=\{(d+1)*m(d+1)-d*m(d)=0, m(1)=1\}:} \\
\MapleInput{recprod:=\{cauchyproduct(rec1, rec2, m(d))=0\}:} \\
\MapleInput{finalrec:=`rec*rec`(recprod, rec3, m(d));}
\MapleOutput{$\{(-2+4*d)*m(d)+(-d-1+hbar^2*d^3+hbar^2*d^2)*m(d+1), m(0) = 0, m(1) = \_C[0]\}$}

The output asserts that
\[
(-2\h+4d\h) \, m(d) + (-d-1+\h^2 d^3 + \h^2 d^2) \, m(d+1) = 0.
\]
By collecting the coefficient of $h^{2g-1}$ and shifting the index, we obtain the 1-point recursion
\[
d \, m_g(d) = 2(2d-3) \, m_g(d-1) + d(d-1)^2 \, m_{g-1}(d).
\]
\end{example}

\section{Examples and applications} \label{sec:applications}

In this section, we return our attention to the enumerative problems introduced in \cref{sec:problems}. In particular, we apply the methodology developed in \cref{sec:recursion} to deduce 1-point recursions for the enumeration of hypermaps, Bousquet-M\'{e}lou--Schaeffer numbers and monotone Hurwitz numbers. For the case of simple Hurwitz numbers, the weight generating function $G(z)$ is not a rational function, so \cref{thm:main} ceases to apply. As a partial converse to this theorem, we show that simple Hurwitz numbers do not satisfy a 1-point recursion. We furthermore demonstrate how our calculations may yield explicit formulas and polynomial structure results for 1-point invariants.

\subsection{Hypermaps and Bousquet-M\'{e}lou--Schaeffer numbers}

The methodology of \cref{sec:recursion} allows one to recover the 1-point recursions for the enumeration of ribbon graphs and dessins d'enfant, stated as \cref{eq:harerzagier,eq:donorbury}, respectively. Recall that these two examples inspired the current work. It is possible to use the methodology developed in \cref{sec:recursion} to deduce other 1-point recursions, although the results are often rather lengthy to state. The following result provides two examples.

\begin{proposition} \label{prop:hypermapsbms}
The 3-hypermap enumeration satisfies the following 1-point recursion.
\begin{align*}
2d(2d+1) \, a^{3}_g(d) ={}& 3 (3d-1) (3d-2) \, a^{3}_g(d-1) + (3d-1) (3d-2) (9d^2-8d+2) \, a^{3}_{g-1}(d-1) \\ 
&- (d-1) (3d-1) (3d-2) (3d-4) (3d-5) (6d-7) \, a^{3}_{g-2}(d-2) \\
&+ (d-1) (d-2) (3d-1) (3d-2) (3d-4) (3d-5) (3d-7) (3d-8) \, a^{3}_{g-3}(d-3)
\end{align*}

The 3-BMS numbers satisfy the following 1-point recursion. 
\begin{align*}
2 d (2d+1) (3d-4) \, b^3_g(d) &= 3 (3d-1) (3d-2) (3d-4) \, b^3_g(d-1) \\
&+ (d-1) (3d+1) (9d^3 - 22d^2 + 14d - 2) \, b^3_{g-1}(d-1) \\
&- (d-1)^2 (d-2) (18d^4 - 93d^3 + 172^2 - 127d + 26) \, b^3_{g-2}(d-2) \\
&+ (d-1)^2 (d-2)^5 (d-3) (3d-1) \, b^3_{g-3}(d-3)
\end{align*}
\end{proposition}

\subsection{Hurwitz numbers}

Observe that \cref{thm:main} does not apply in the case of Hurwitz numbers, since the weight generating function $G(z) = \exp(z)$ is not rational. Thus, the following result provides a partial converse to our main theorem.

\begin{proposition}
The simple Hurwitz numbers do not satisfy a 1-point recursion.
\end{proposition}

\begin{proof}
By \cref{lem:holonomic}, we know that the simple Hurwitz numbers satisfy a 1-point recursion if and only if the sequence
\[
h(d) = \frac{1}{d! \, \h^d} \sum_{k=1}^d (-1)^{d-k} \binom{d-1}{k-1} \exp(d (2k-d-1) \h / 2) = \frac{1}{d! \, \h^d} \exp(-d(d+1)\h/2) \, (\exp(d\h)-1)^{d-1} 
\]
is holonomic over $\mathbb{C}(\h)$. However, if this were the case, then we could evaluate at $\h = 1$ to deduce that the sequence
\[
\frac{1}{d!} \exp(-d(d+1)/2) \, (\exp(d)-1)^{d-1} 
\]
is holonomic over $\mathbb{C}$. It is known that holonomic sequences $a_1, a_2, a_3, \ldots$ over $\mathbb{C}$ must satisfy the asymptotic growth condition $a_d = O(d!^\alpha)$ for some constant $\alpha$. On the other hand, we have
\[
\frac{1}{d!} \exp(-d(d+1)/2) \, (\exp(d)-1)^{d-1} \sim \frac{1}{d!} \exp(d(d-3)/2).
\]
Applying Stirling's formula, we see that this grows too fast to be holonomic. So it follows that the simple Hurwitz numbers do not satisfy a 1-point recursion.
\end{proof}

\Cref{{eq:1pointgf}} still applies to this case though, so the 1-part Hurwitz numbers satisfy
\[
\sum_{g=0}^\infty h_g(d) \,\h^{2g-1} = \frac{1}{d! \, \h^d} \sum_{k=1}^d (-1)^{d-k} \binom{d-1}{k-1} \exp \left( d (2k-d-1) \h/2 \right).
\]
By extracting coefficients of $\h$ on both sides, we recover the following formula.

\begin{proposition} \label{prop:1parthurwitz}
The 1-part simple Hurwitz numbers are given by
\[
h_g(d) = \frac{(d/2)^{d+2g-1}}{d! \, (d+2g-1)!}\sum_{k=0}^{d-1} (-1)^k \, \binom{d-1}{k} \, (d-1-2k)^{d+2g-1}.
\]
In particular, it follows that $h_g(d) = \frac{d^d}{d!} p_g(d)$, where $p_g$ is a polynomial of degree $3g-1$. One can make sense of this statement in the case $g=0$ by taking $p_0(d) = \frac{1}{d}$.

\end{proposition}

We remark that the polynomial structure derived here is a direct corollary of the more general polynomial structure for simple Hurwitz numbers with any number of parts. This in turn follows from the ELSV formula, which relates simple Hurwitz numbers to intersection theory on moduli spaces of curves~\cite{eke-lan-sha-vai}. The formula of \cref{prop:1parthurwitz} is not new either, but first appeared in the work of Shapiro, Shapiro and Vainshtein~\cite{sha-sha-vai}. The result and proof here may generalise to other settings, as we will observe in the context of monotone Hurwitz numbers.

\subsection{Monotone Hurwitz numbers}

In \cref{subsec:algorithm}, we observed that the following 1-point recursion for monotone Hurwitz numbers could be deduced from several lines of {\sc Maple} code. As with the Harer--Zagier recursion, it would be of interest to have an independent and purely combinatorial proof of this statement.

\begin{proposition} \label{prop:monotonehurwitz}
The 1-part monotone Hurwitz numbers satisfy the 1-point recursion
\[
d \, m_g(d) = 2(2d-3) \, m_g(d-1) + d(d-1)^2 \, m_{g-1}(d).
\]
\end{proposition}

In the context of monotone Hurwitz numbers, \cref{eq:1pointgf} implies that
\begin{align*}
\sum_{g=0}^\infty m_g(d) \,\h^{2g-1} &= \frac{1}{d! \, h^d} \sum_{k=1}^d (-1)^{d-k} \binom{d-1}{k-1} \prod_{j=1}^d \frac{1}{1-(k-j)\h} \\
&= \frac{(2d-2)!}{d! \, (d-1)!} \prod_{k=-d+1}^{d-1} \frac{1}{1-k\h}.
\end{align*}
The identity that leads to the second equality can be established by considering the residue at $\h = \frac{1}{k}$ for $-d+1 \leq k \leq d-1$. By extracting coefficients of $\h$ on both sides, we recover the following formula.

\begin{corollary} \label{cor:monotonehurwitz}
The 1-part monotone Hurwitz numbers satisfy the equation
\begin{align*}
m_g(d) &= \frac{(2d-2)!}{d! (d-1)!} \sum_{k_1 + \cdots + k_{d-1}=g} \prod_{i=1}^{d-1} i^{2k_i} \\
&= \frac{(2d-2)!}{d! (d-1)!} \sum_{1 \leq m_1 \leq m_2 \leq \cdots \leq m_g \leq d-1} (m_1 m_2 \cdots m_g)^2.
\end{align*}
From the latter summation, it follows that $m_g(d) = \binom{2d}{d} \, p_g(d)$, where $p_g$ is a polynomial of degree $3g-1$. One can make sense of this statement in the case $g=0$ by taking $p_0(d) = \frac{1}{d}$.
\end{corollary}

This polynomial structure is a particular case of the more general result of Goulden, Guay-Paquet and Novak~\cite{gou-gua-nov13b}, who prove that monotone Hurwitz numbers satisfy
\[
M_{g,n}(d_1, d_2, \ldots, d_n) = \prod_{i=1}^n \binom{2d_i}{d_i} \times P_{g,n}(d_1, d_2, \ldots, d_n),
\]
where $P_{g,n}$ is a polynomial of degree $3g-3+n$. One wonders whether the techniques of this paper can be used to prove this more general structure theorem.

\section{Relations to topological recursion and quantum curves} \label{sec:discussion}

\subsection{Topological recursion}

In this section, we aim to address the question: how universal is the the notion of a 1-point recursion? Thus, one seeks a natural class of ``enumerative'' problems for which 1-point recursions exist. Such a class should include not only the ribbon graph and dessin d'enfant enumerations, but also those families of problems encompassed by \cref{thm:main2} --- namely, those arising from the double Schur expansion of \cref{eq:partitionfunction} with $\mathbf{q} = (q_1, q_2, \ldots, q_r, 0, 0, \ldots)$ and a rational weight generating function $G(z)$. We claim that a natural candidate is the class of problems governed by the topological recursion that we subsequently discuss.

The topological recursion of Chekhov, Eynard and Orantin was originally inspired by the loop equations in the theory of matrix models~\cite{che-eyn, eyn-ora07}. It has since found widespread applications to various problems across mathematics and physics. For example, it is known to govern the enumeration of maps on surfaces~\cite{and-che-nor-pen, do-man, dum-mul-saf-sor, dun-ora-pop-sha, kaz-zog, nor}, various flavours of Hurwitz problems~\cite{bou-her-liu-mul, bou-mar, do-dye-mat, do-lei-nor, eyn-mul-saf}, the Gromov--Witten theory of $\mathbb{P}^1$~\cite{dun-ora-sha-spi, nor-sco} and toric Calabi--Yau threefolds~\cite{bou-kle-mar-pas, eyn-ora15, fan-liu-zon}. There are also conjectural relations to knots invariants~\cite{bor-eyn, gu-joc-kle-sor}. Much of the power of the topological recursion lies in its universality --- in other words, its wide applicability across broad classes of problems --- and its ability to reveal commonality among such problems.

The topological recursion can naively be thought of as a vast generalisation of Tutte's recursion for the enumeration of ribbon graphs. It calculates $n$-point functions in a recursive manner, starting from the input data of a spectral curve. For our purposes, we restrict to the class of {\em rational spectral curves}, that are given by a pair $(x(z), y(z))$ of rational functions satisfying some mild assumptions. For more information on the theory of the topological recursion, one should consult the relevant literature~\cite{eyn-ora07}.

The following result asserts that the weighted Hurwitz numbers --- essentially, the $N_{g,n}(d_1, d_2, \ldots, d_n)$ of \cref{eq:partitionfunction} --- are governed by the topological recursion.

\begin{theorem}[Alexandrov, Chapuy, Eynard and Harnad~\cite{ale-cha-eyn-har}] The rational spectral curve given by
\[
x(z) = \frac{z}{G(Q(z))} \qquad \text{and} \qquad y(z) = \frac{Q(z)}{z} G(Q(z)), \qquad \text{where } Q(z) = q_1 z + q_2 z^2 + \cdots + q_r z^r,
\]
produces correlation differentials that satisfy
\[
\omega_{g,n} = \sum_{d_1, d_2, \ldots, d_n=1}^\infty N_{g,n}(d_1, d_2, \ldots, d_n) \prod_{i=1}^n d_i x_i^{d_i-1} \, \dd x_i.
\]
\end{theorem}

This lends credence to the following conjecture, which states that 1-point recursions exist for rational spectral curves in general.

\begin{conjecture} \label{con:mainprecise}
Consider a rational spectral curve given by the pair of rational functions $(x(z), y(z))$. Suppose that the correlation differentials produced by the topological recursion applied to this spectral curve have an expansion of the form
\[
\omega_{g,n} = \sum_{d_1, d_2, \ldots, d_n = 1}^\infty N_{g,n}(d_1, d_2, \ldots, d_n) \prod_{i=1}^n d_i x_i^{d_i-1} \, \dd x_i.
\]
Then the numbers $n_g(d) = d \, N_{g,1}(d)$ satisfy a 1-point recursion.
\end{conjecture}

We conclude this section with an example of a problem that is governed by topological recursion and satisfies a 1-point recursion, but does not satisfy the conditions of \cref{thm:main2}. Thus, one can consider this as further evidence towards the conjecture above.

\begin{example}
Chekhov and Norbury~\cite{che-nor} consider topological recursion applied to the spectral curve $x^2y^2 - 4y^2 - 1 = 0$ given by the rational parametrisation
\[
x(z) = z + \frac{1}{z} \qquad \text{and} \qquad y(z) = \frac{z}{z^2-1}.
\]
The resulting correlation differentials can be expressed as
\[
\omega_{g,n} = \sum_{d_1, d_2, \ldots, d_n = 1}^\infty J_{g,n}(d_1, d_2, \ldots, d_n) \prod_{i=1}^n d_i z_i^{d_i-1} \, \dd z_i.
\]
These are derivatives of the correlation functions for the Legendre ensemble, which arise from a particular Hermitian matrix model, as well as related models from conformal field theory. In the latter context, Gaberdiel, Klemm and Runkel use null vectors for Virasoro highest weight representations to deduce an equation \cite[equation (4.18)]{gab-kle-run} that is equivalent to a 1-point recursion for the numbers $j_g(d) = d \, J_{g,1}(d)$. In summary, the 1-point invariants produced by the topological recursion on the rational spectral curve above satisfy a 1-point recursion.\footnote{Observe that we are here expanding in $z$, while \cref{con:mainprecise} has been expressed in terms of~$x$. However, since they are related by a rational change of coordinates, this does not affect the existence of a 1-point recursion.}
\end{example}

Kontsevich and Soibelman have recently provided an alternative and more general formulation of the topological recursion~\cite{kon-soi}. It allows one to calculate $n$-point functions using a technique that is ostensibly more algebraic and less analytic. So it may provide a promising approach to \cref{con:mainprecise}.

\subsection{Quantum curves}

The notion of quantum curves is closely related to that of topological recursion~\cite{nor15}. In short, they are non-commutative deformations of spectral curves that are used as the input to the topological recursion. Although it is not currently clear when they exist, the quantum curve phenomenon has been proven or observed in many instances of the topological recursion.

A quantum curve can be viewed as a differential operator $\widehat{P}(\widehat{x}, \widehat{y})$ that annihilates the so-called principal specialisation of the partition function.
\[
\widehat{P}(\widehat{x}, \widehat{y}) \left. Z(\mathbf{p}; \h) \right|_{p_i=x^i} = 0
\]
We use here the operators $\widehat{x} = x$ and $\widehat{y} = -\h \frac{\partial}{\partial x}$. The quantum curve phenomenon is the fact that there is a natural choice of the operator $\widehat{P}(\widehat{x}, \widehat{y})$ whose semi-classical limit --- obtained by setting $\h = 0$ and allowing $x$ and $y$ to commute --- recovers the spectral curve $P(x, y) = 0$.

In the context of the double Schur expansions considered in this paper, the principle specialisation of the wave function is given by
\[
\Psi(x; \mathbf{q}; \h) = \sum_{\lambda \in {\mathcal P}} s_\lambda(x, x^2, x^3, \ldots) \, s_\lambda(\tfrac{q_1}{\h}, \tfrac{q_2}{\h}, \ldots) \, \prod_{\Box \in \lambda} G(c(\Box) \h).
\]
As in \cref{sec:schur}, the hook-content formula stated in~\cref{eq:hookcontent} may be invoked to simplify the expression to obtain
\[
\Psi(x; \mathbf{q}; \h) = \sum_{d=0}^\infty x^d \, s_{(d)}(\tfrac{q_1}{\h}, \tfrac{q_2}{\h}, \ldots) \, \prod_{k=1}^{d-1} G(k\h) = \sum_{d=0}^\infty x^d \, \prod_{k=1}^{d-1} G(k\h) [y^d] \exp \bigg( \sum_{k=1}^r \frac{q_k}{k\h} y^k \bigg).
\]
Here, $[y^d]$ denotes extraction of the coefficient of $y^d$.

We simply remark here that our calculation of the 1-point invariants from the partition function in \cref{sec:schur} bears a strong resemblance to the calculation of the quantum curve from the partition function~\cite{ale-cha-eyn-har, ale-lew-sha, mul-sha-spi}. In the former case, the partition function reduces to a sum over hook partitions, while in the latter case, it reduces to a sum over 1-part partitions. One may wonder whether there may be a deeper connection here.

\begin{small}
\bibliographystyle{plain}
\bibliography{1-point-recursions}

\begin{thebibliography}{10}

\bibitem{akh-sha}
\`E.~T. Akhmedov and Sh.~R. Shakirov.
\newblock Gluings of surfaces with polygonal boundaries.
\newblock {\em Funktsional. Anal. i Prilozhen.}, 43(4):3--13, 2009.

\bibitem{ale-cha-eyn-har}
A.~Alexandrov, G.~Chapuy, B.~Eynard, and J.~Harnad.
\newblock Fermionic approach to weighted {H}urwitz numbers and topological
  recursion.
\newblock {\em Comm. Math. Phys.}, 360(2):777--826, 2018.

\bibitem{ale-lew-sha}
A.~Alexandrov, D.~Lewanski, and S.~Shadrin.
\newblock Ramifications of {H}urwitz theory, {KP} integrability and quantum
  curves.
\newblock {\em J. High Energy Phys.}, (5):124, front matter+30, 2016.

\bibitem{and-che-nor-pen}
J\o rgen~Ellegaard Andersen, Leonid~O. Chekhov, Paul Norbury, and Robert~C.
  Penner.
\newblock Models of discretized moduli spaces, cohomological field theories,
  and {G}aussian means.
\newblock {\em J. Geom. Phys.}, 98:312--339, 2015.

\bibitem{bor-eyn}
Ga\"etan Borot and Bertrand Eynard.
\newblock All order asymptotics of hyperbolic knot invariants from
  non-perturbative topological recursion of {A}-polynomials.
\newblock {\em Quantum Topol.}, 6(1):39--138, 2015.

\bibitem{bou-her-liu-mul}
Vincent Bouchard, Daniel Hern\'andez~Serrano, Xiaojun Liu, and Motohico Mulase.
\newblock Mirror symmetry for orbifold {H}urwitz numbers.
\newblock {\em J. Differential Geom.}, 98(3):375--423, 2014.

\bibitem{bou-kle-mar-pas}
Vincent Bouchard, Albrecht Klemm, Marcos Mari\~no, and Sara Pasquetti.
\newblock Remodeling the {B}-model.
\newblock {\em Comm. Math. Phys.}, 287(1):117--178, 2009.

\bibitem{bou-mar}
Vincent Bouchard and Marcos Mari\~no.
\newblock Hurwitz numbers, matrix models and enumerative geometry.
\newblock In {\em From {H}odge theory to integrability and {TQFT}
  tt*-geometry}, volume~78 of {\em Proc. Sympos. Pure Math.}, pages 263--283.
  Amer. Math. Soc., Providence, RI, 2008.

\bibitem{bou-sch}
Mireille Bousquet-M\'elou and Gilles Schaeffer.
\newblock Enumeration of planar constellations.
\newblock {\em Adv. in Appl. Math.}, 24(4):337--368, 2000.

\bibitem{cha-fer-fus}
Guillaume Chapuy, Valentin F\'eray, and \'Eric Fusy.
\newblock A simple model of trees for unicellular maps.
\newblock {\em J. Combin. Theory Ser. A}, 120(8):2064--2092, 2013.

\bibitem{che-eyn}
Leonid Chekhov and Bertrand Eynard.
\newblock Hermitian matrix model free energy: {F}eynman graph technique for all
  genera.
\newblock {\em J. High Energy Phys.}, (3):014, 18, 2006.

\bibitem{che-nor}
Leonid Chekhov and Paul Norbury.
\newblock Topological recursion with hard edges, 2017.

\bibitem{che}
Leonid~O. Chekhov.
\newblock The {H}arer-{Z}agier recursion for an irregular spectral curve.
\newblock {\em J. Geom. Phys.}, 110:30--43, 2016.

\bibitem{do-kar}
N.~Do and M.~Karev.
\newblock Monotone orbifold {H}urwitz numbers.
\newblock {\em Zap. Nauchn. Sem. S.-Peterburg. Otdel. Mat. Inst. Steklov.
  (POMI)}, 446(Kombinatorika i Teoriya Grafov. V):40--69, 2016.

\bibitem{do-dye-mat}
Norman Do, Alastair Dyer, and Daniel~V. Mathews.
\newblock Topological recursion and a quantum curve for monotone {H}urwitz
  numbers.
\newblock {\em J. Geom. Phys.}, 120:19--36, 2017.

\bibitem{do-kar18}
Norman Do and Maksim Karev.
\newblock Towards the topological recursion for double {H}urwitz numbers, 2018.

\bibitem{do-lei-nor}
Norman Do, Oliver Leigh, and Paul Norbury.
\newblock Orbifold {H}urwitz numbers and {E}ynard-{O}rantin invariants.
\newblock {\em Math. Res. Lett.}, 23(5):1281--1327, 2016.

\bibitem{do-man}
Norman Do and David Manescu.
\newblock Quantum curves for the enumeration of ribbon graphs and hypermaps.
\newblock {\em Commun. Number Theory Phys.}, 8(4):677--701, 2014.

\bibitem{do-nor}
Norman Do and Paul Norbury.
\newblock Topological recursion for irregular spectral curves.
\newblock {\em J. Lond. Math. Soc. (2)}, 97(3):398--426, 2018.

\bibitem{dum-mul-saf-sor}
Olivia Dumitrescu, Motohico Mulase, Brad Safnuk, and Adam Sorkin.
\newblock The spectral curve of the {E}ynard-{O}rantin recursion via the
  {L}aplace transform.
\newblock In {\em Algebraic and geometric aspects of integrable systems and
  random matrices}, volume 593 of {\em Contemp. Math.}, pages 263--315. Amer.
  Math. Soc., Providence, RI, 2013.

\bibitem{dun-ora-sha-spi}
P.~Dunin-Barkowski, N.~Orantin, S.~Shadrin, and L.~Spitz.
\newblock Identification of the {G}ivental formula with the spectral curve
  topological recursion procedure.
\newblock {\em Comm. Math. Phys.}, 328(2):669--700, 2014.

\bibitem{dun-ora-pop-sha}
Petr Dunin-Barkowski, Nicolas Orantin, Aleksandr Popolitov, and Sergey Shadrin.
\newblock Combinatorics of loop equations for branched covers of sphere.
\newblock {\em Int. Math. Res. Not. IMRN}, (18):5638--5662, 2018.

\bibitem{eke-lan-sha-vai}
Torsten Ekedahl, Sergei Lando, Michael Shapiro, and Alek Vainshtein.
\newblock Hurwitz numbers and intersections on moduli spaces of curves.
\newblock {\em Invent. Math.}, 146(2):297--327, 2001.

\bibitem{eyn-ora07}
B.~Eynard and N.~Orantin.
\newblock Invariants of algebraic curves and topological expansion.
\newblock {\em Commun. Number Theory Phys.}, 1(2):347--452, 2007.

\bibitem{eyn-ora15}
B.~Eynard and N.~Orantin.
\newblock Computation of open {G}romov-{W}itten invariants for toric
  {C}alabi-{Y}au 3-folds by topological recursion, a proof of the {BKMP}
  conjecture.
\newblock {\em Comm. Math. Phys.}, 337(2):483--567, 2015.

\bibitem{eyn-mul-saf}
Bertrand Eynard, Motohico Mulase, and Bradley Safnuk.
\newblock The {L}aplace transform of the cut-and-join equation and the
  {B}ouchard-{M}ari\~no conjecture on {H}urwitz numbers.
\newblock {\em Publ. Res. Inst. Math. Sci.}, 47(2):629--670, 2011.

\bibitem{fan-liu-zon}
Bohan Fang, Chiu-Chu~Melissa Liu, and Zhengyu Zong.
\newblock The {E}ynard-{O}rantin recursion and equivariant mirror symmetry for
  the projective line.
\newblock {\em Geom. Topol.}, 21(4):2049--2092, 2017.

\bibitem{gab-kle-run}
Matthias~R. Gaberdiel, Albrecht Klemm, and Ingo Runkel.
\newblock Matrix model eigenvalue integrals and twist fields in the su (2)-wzw
  model.
\newblock {\em Journal of High Energy Physics}, 2005(10):107, 2005.

\bibitem{gou-gua-nov13a}
I.~P. Goulden, Mathieu Guay-Paquet, and Jonathan Novak.
\newblock Monotone {H}urwitz numbers in genus zero.
\newblock {\em Canad. J. Math.}, 65(5):1020--1042, 2013.

\bibitem{gou-gua-nov13b}
I.~P. Goulden, Mathieu Guay-Paquet, and Jonathan Novak.
\newblock Polynomiality of monotone {H}urwitz numbers in higher genera.
\newblock {\em Adv. Math.}, 238:1--23, 2013.

\bibitem{gou-gua-nov14}
I.~P. Goulden, Mathieu Guay-Paquet, and Jonathan Novak.
\newblock Monotone {H}urwitz numbers and the {HCIZ} integral.
\newblock {\em Ann. Math. Blaise Pascal}, 21(1):71--89, 2014.

\bibitem{gou-nic}
I.~P. Goulden and A.~Nica.
\newblock A direct bijection for the {H}arer-{Z}agier formula.
\newblock {\em J. Combin. Theory Ser. A}, 111(2):224--238, 2005.

\bibitem{gu-joc-kle-sor}
Jie Gu, Hans Jockers, Albrecht Klemm, and Masoud Soroush.
\newblock Knot invariants from topological recursion on augmentation varieties.
\newblock {\em Comm. Math. Phys.}, 336(2):987--1051, 2015.

\bibitem{har-zag}
J.~Harer and D.~Zagier.
\newblock The {E}uler characteristic of the moduli space of curves.
\newblock {\em Invent. Math.}, 85(3):457--485, 1986.

\bibitem{hur}
A.~Hurwitz.
\newblock Ueber {R}iemann'sche {F}l\"achen mit gegebenen {V}erzweigungspunkten.
\newblock {\em Math. Ann.}, 39(1):1--60, 1891.

\bibitem{kau-pau}
Manuel Kauers and Peter Paule.
\newblock {\em The concrete tetrahedron}.
\newblock Texts and Monographs in Symbolic Computation. SpringerWienNewYork,
  Vienna, 2011.
\newblock Symbolic sums, recurrence equations, generating functions, asymptotic
  estimates.

\bibitem{kaz-zog}
Maxim Kazarian and Peter Zograf.
\newblock Virasoro constraints and topological recursion for {G}rothendieck's
  dessin counting.
\newblock {\em Lett. Math. Phys.}, 105(8):1057--1084, 2015.

\bibitem{kon-soi}
Maxim Kontsevich and Yan Soibelman.
\newblock {A}iry structures and symplectic geometry of topological recursion,
  2017.

\bibitem{lan-zvo}
Sergei~K. Lando and Alexander~K. Zvonkin.
\newblock {\em Graphs on surfaces and their applications}, volume 141 of {\em
  Encyclopaedia of Mathematical Sciences}.
\newblock Springer-Verlag, Berlin, 2004.
\newblock With an appendix by Don B. Zagier, Low-Dimensional Topology, II.

\bibitem{las}
Bodo Lass.
\newblock D\'emonstration combinatoire de la formule de {H}arer-{Z}agier.
\newblock {\em C. R. Acad. Sci. Paris S\'er. I Math.}, 333(3):155--160, 2001.

\bibitem{led}
M.~Ledoux.
\newblock A recursion formula for the moments of the {G}aussian orthogonal
  ensemble.
\newblock {\em Ann. Inst. Henri Poincar\'e Probab. Stat.}, 45(3):754--769,
  2009.

\bibitem{MR929767}
L.~Lipshitz.
\newblock The diagonal of a {$D$}-finite power series is {$D$}-finite.
\newblock {\em J. Algebra}, 113(2):373--378, 1988.

\bibitem{mac}
I.G. Macdonald.
\newblock {\em Symmetric Functions and Hall Polynomials}.
\newblock Oxford classic texts in the physical sciences. Clarendon Press, 1998.

\bibitem{mor-sha}
A.~Morozov and Sh. Shakirov.
\newblock From {B}rezin--{H}ikami to {H}arer---{Z}agier formulas for {G}aussian
  correlators, 2010.

\bibitem{mul-sha-spi}
M.~Mulase, S.~Shadrin, and L.~Spitz.
\newblock The spectral curve and the {S}chr\"odinger equation of double
  {H}urwitz numbers and higher spin structures.
\newblock {\em Commun. Number Theory Phys.}, 7(1):125--143, 2013.

\bibitem{nor}
Paul Norbury.
\newblock String and dilaton equations for counting lattice points in the
  moduli space of curves.
\newblock {\em Trans. Amer. Math. Soc.}, 365(4):1687--1709, 2013.

\bibitem{nor15}
Paul Norbury.
\newblock Quantum curves and topological recursion.
\newblock In {\em String-{M}ath 2014}, volume~93 of {\em Proc. Sympos. Pure
  Math.}, pages 41--65. Amer. Math. Soc., Providence, RI, 2016.

\bibitem{nor-sco}
Paul Norbury and Nick Scott.
\newblock Gromov-{W}itten invariants of {$\mathbb{P}^1$} and {E}ynard-{O}rantin
  invariants.
\newblock {\em Geom. Topol.}, 18(4):1865--1910, 2014.

\bibitem{oko-pan}
A.~Okounkov and R.~Pandharipande.
\newblock Gromov-{W}itten theory, {H}urwitz theory, and completed cycles.
\newblock {\em Ann. of Math. (2)}, 163(2):517--560, 2006.

\bibitem{oko}
Andrei Okounkov.
\newblock Toda equations for {H}urwitz numbers.
\newblock {\em Math. Res. Lett.}, 7(4):447--453, 2000.

\bibitem{orl-shc}
A.~Yu. Orlov and D.~M. Shcherbin.
\newblock Hypergeometric solutions of soliton equations.
\newblock {\em Teoret. Mat. Fiz.}, 128(1):84--108, 2001.

\bibitem{pit}
Boris Pittel.
\newblock Another proof of the {H}arer-{Z}agier formula.
\newblock {\em Electron. J. Combin.}, 23(1):Paper 1.21, 11, 2016.

\bibitem{sal-zim}
Bruno Salvy and Paul Zimmermann.
\newblock Gfun: A maple package for the manipulation of generating and
  holonomic functions in one variable.
\newblock {\em ACM Trans. Math. Softw.}, 20(2):163--177, June 1994.

\bibitem{sha-sha-vai}
B.~Shapiro, M.~Shapiro, and A.~Vainshtein.
\newblock Ramified coverings of {$S^2$} with one degenerate branch point and
  enumeration of edge-ordered graphs.
\newblock In {\em Topics in singularity theory}, volume 180 of {\em Amer. Math.
  Soc. Transl. Ser. 2}, pages 219--227. Amer. Math. Soc., Providence, RI, 1997.

\bibitem{MR587530}
R.~P. Stanley.
\newblock Differentiably finite power series.
\newblock {\em European J. Combin.}, 1(2):175--188, 1980.

\bibitem{tut}
W.~T. Tutte.
\newblock A census of planar maps.
\newblock {\em Canad. J. Math.}, 15:249--271, 1963.

\bibitem{wal-leh}
T.~Walsh and A.~B. Lehman.
\newblock Counting rooted maps by genus. {II}.
\newblock {\em J. Combinatorial Theory Ser. B}, 13:122--141, 1972.

\end{thebibliography}
\end{small}

\textsc{School of Mathematical Sciences, Monash University, VIC 3800, Australia} \\
\emph{Email:} \href{mailto:anupam.chaudhuri@monash.edu}{anupam.chaudhuri@monash.edu}

\textsc{School of Mathematical Sciences, Monash University, VIC 3800, Australia} \\
\emph{Email:} \href{mailto:norm.do@monash.edu}{norm.do@monash.edu}

\end{document}